\documentclass[a4paper, runningheads, orivec]{llncs}

 \def \VersionForArXiV {}

\usepackage[T1]{fontenc}

\usepackage{graphicx}

\newif\iffinal
\finaltrue

\newif\iffull
\fulltrue

\makeatletter
\AtBeginDocument{%
  \@ifpackageloaded{hyperref}
  {\def\@doi#1{\href{https://doi.org/#1}
      {\ttfamily https://doi.org/#1}\egroup}}
  {\def\@doi#1{\ttfamily https://doi.org/#1\egroup}}
  \def\doi{\bgroup\catcode`\_=12\relax\@doi}}
\makeatother

\newcommand{\LongVersion}[1]{}

\usepackage[svgnames,table]{xcolor}
\definecolor{darkblue}{rgb}{0, 0, 0.7}

\usepackage[
		pdfauthor={É. André, P. Eichler, S. Jacobs and S. Karra},%
		pdftitle={Parameterized Verification of Disjunctive Timed Networks},
		breaklinks  = true,
		colorlinks  = true,
	\ifdefined \VersionWithComments
	\fi
		citecolor   = blue!50!blue,
		linkcolor   = darkblue,
		urlcolor    = blue!50!green,
	]{hyperref}

\renewcommand{\orcidID}[1]{\href{#1}{\includegraphics[height=9pt]{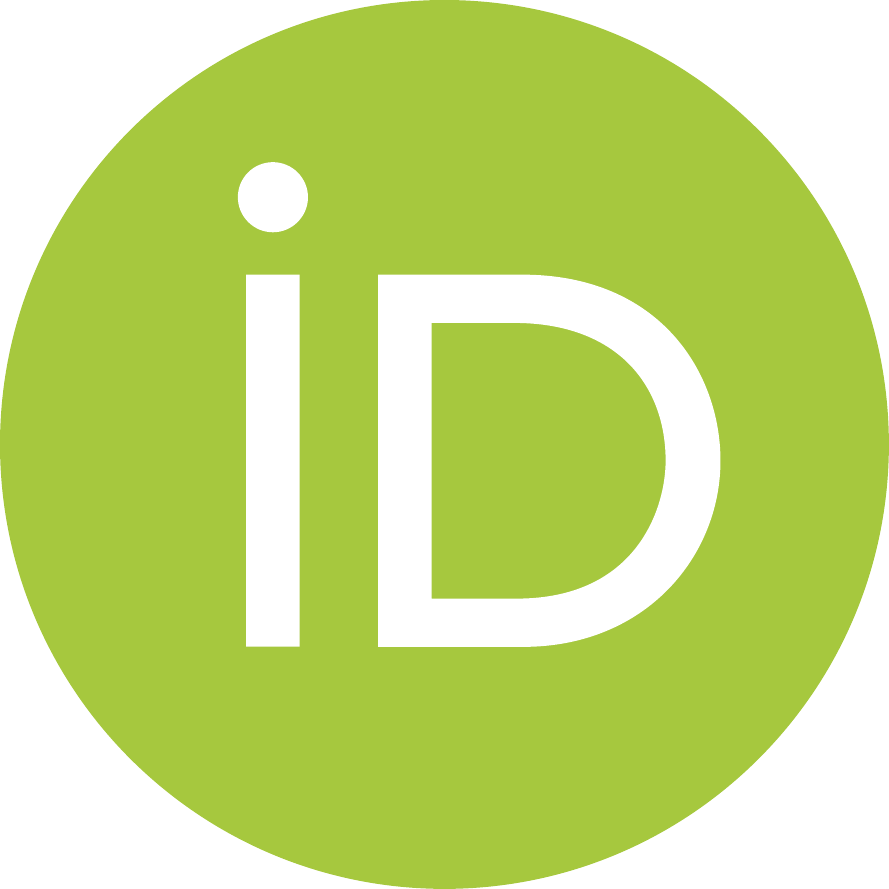}}}
\usepackage{fontawesome}
\newcommand{\homepage}[1]{\href{#1}{\color{gray}\faHome}}

\usepackage{caption}
\usepackage{subcaption}
\usepackage{graphicx}
\usepackage{amsfonts}
\usepackage{amssymb}
\usepackage{amsmath}
\usepackage{cancel}
\usepackage{mathtools}
\usepackage{tikz}
\usepackage[noend,linesnumbered,ruled]{algorithm2e}
\usepackage{upgreek}
\usepackage{xspace}
\usepackage{cite}
\usepackage{multirow}
\usepackage{booktabs}

\usepackage{thmtools}
\usepackage{thm-restate}

\usepackage[capitalise,english,nameinlink]{cleveref} %
\crefname{line}{\text{line}}{\text{lines}} %

\newcommand{\eg}{e.g.,\xspace}

\newcommand{\ie}{i.e.,\xspace}

\usepackage{paralist} %

\newenvironment{ienumerate}
	{\begin{inparaenum}[\itshape i\upshape)]}
	{\end{inparaenum}}

\newenvironment{oneenumerate}
	{\ifdefined\VersionLong\begin{enumerate}\else\begin{inparaenum}[1)]\fi}
	{\ifdefined\VersionLong\end{enumerate}\else\end{inparaenum}\fi}

\iffinal
  
\else
  
\fi

\newcommand{\card}[1]{\ensuremath{\left|#1\right|}}

\usepackage{blindtext}

\usepackage{floatrow}

\makeatletter
\def\BState{\State\hskip-\ALG@thistlm}
\makeatother
\tikzset{
     redblock/.style={rectangle, fill=red!40, text width=2em,
                   text centered, rounded corners, minimum height=1em},
     state/.style={circle, text width=2em,
                   text centered, minimum height=2em},
     basic/.style  = {draw, text width=2cm,  font=\sffamily, rectangle},
     arrow/.style={-{Stealth[]}}
     }

\newfloatcommand{capbtabbox}{table}[][\FBwidth]
\newcommand{\transition} [3]{{#1}\xrightarrow{#2}{#3}} 
\newcommand{\loc}{\ensuremath{q}}
\newcommand{\Loc}{\ensuremath{Q}}

\newcommand{\smartpar}[1]{\medskip \noindent \textbf{#1.}}
\newcommand{\gclock}{\ensuremath{t}}
\newcommand{\minval}[2]{\ensuremath{#1_{min}(#2)}\xspace}
\newcommand{\minvalnet}[3]{\ensuremath{#2_{min}^{#1}(#3)}\xspace}
\newcommand{\totaltime}{\ensuremath{\delta}}
\newcommand{\vmin}{\ensuremath{v_{min}}}

\newcommand{\sguard}[1]{\ensuremath{\mathit{lguard}(#1)}}

\newcommand{\LocGuards}{\ensuremath{\mathit{Guards}}}
\newcommand{\RGuards}{\ensuremath{\mathit{RGuards}}}

\newcommand{\ReachLoc}[1]{\ensuremath{\mathit{ReachL}(#1)}} %
\newcommand{\UG}[1]{\ensuremath{\mathit{UG}(#1)}} %

\newcommand{\invariant}{\ensuremath{\mathit{Inv}}}
\newcommand{\intervalPr}{\ensuremath{\mathcal{I}}}

\usepackage{accents}

\newcommand{\lcomputation}{\ensuremath{\rho}}
\newcommand{\gcomputation}{\ensuremath{\pi}}

\newcommand{\Lg}{\mathcal{L}}

\newcommand{\mQz}{\ensuremath{\mathbb{R}_{\geq 0}}}
\newcommand{\Nats}{\ensuremath{\mathbb{N}}}
\newcommand{\Ints}{\ensuremath{\mathbb{Z}}}

\newcommand{\famTemplates}{\ensuremath{\mathcal{A}}}

\newcommand{\gta}{gTA\xspace}
\newcommand{\Au}{\ensuremath{A'}\xspace}
\newcommand{\summary}[1]{\hat{#1}}
\newcommand{\sympath}{\ensuremath{\overline{\lcomputation}}}

\newcommand{\ZoneGraph}{\ensuremath{ZG}}

\newcommand{\zgDTN}{\ensuremath{\ZoneGraph^\infty(A')}\xspace}
\newcommand{\bzgDTN}{\ensuremath{\ZoneGraph_{\UpperBound(A)}^\infty(A')}\xspace}

\newcommand{\zone}{\ensuremath{z}}

\newcommand{\asap}{\ensuremath{\lcomputation_\mathsf{asap}}}
\newcommand{\finalq}{\ensuremath{\mathsf{lastloc}}}
\newcommand{\finalc}{\ensuremath{\mathsf{lastconf}}}
\newcommand{\firstq}{\ensuremath{\mathsf{firstloc}}}

\newcommand{\minreach}{\textsc{Minreach}\xspace}
\newcommand{\reach}{\ensuremath{\mathsf{reach}}}

\newcommand{\reachorder}{\ensuremath{\leq_\reach}}

\newcommand{\UpperBound}{\ensuremath{\mathit{UB}}}
\newcommand{\timeBound}{\ensuremath{B}}

\newcommand{\minReachT}[1]{\ensuremath{\minval{\delta}{#1}}} %
\newcommand{\niMinReachT}[2]{\ensuremath{\minvalnet{#1}{\delta}{#2}}} %
\newcommand{\nInfMinReachT}[1]{\ensuremath{\minvalnet{\infty}{\delta}{#1}}} %

\newcommand{\DTNpersistent}{\DTN{}$^-$\xspace}
\newcommand{\DTN}{DTN\xspace}

\newcommand{\Uppaal}{\textsc{Uppaal}}

\newcommand{\nConfig}{\ensuremath{\mathfrak{c}}} %
\newcommand{\nConfigSet}{\ensuremath{\mathfrak{C}}}
\newcommand{\nConfigInit}{\ensuremath{\hat{\nConfig}}} %
\newcommand{\nDisTrans}[1]{\ensuremath{(i_{#1},\tau_{#1})}} %

 \newcommand{\psione}{q_0 \xrightarrow[]{\tau_0}\ldots \xrightarrow[]{\tau_{i}}q_{i+1}}
 \newcommand{\psitwo}{q_{i+1}\xrightarrow[]{\tau_{i+1}} \ldots \xrightarrow[]{\tau_j} q_{j+1}}
 \newcommand{\psithree}{q_{j+1} \xrightarrow[]{\tau_{j+1}} \ldots \xrightarrow[]{\tau_{l-1}} q_l}

\newcommand{\ourloop}{(q_0,u_0) \xrightarrow{\delta_{0},\tau_{0}} \ldots \xrightarrow{\delta_{l-1},\tau_{l-1}} (q_l,u_l)}

\usepackage{blindtext}

\newcommand{\assign}{\leftarrow}

\newcommand{\clock}{\ensuremath{x}}
\newcommand{\clocks}{\ensuremath{\mathsf{C}}}
\newcommand{\clockcons}{\ensuremath{\mathcal{CC}(\clocks)}}
\newcommand{\ccons}{\ensuremath{\varphi}}
\newcommand{\intconstant}{\ensuremath{d}}

\usepackage[most]{tcolorbox}

\makeatletter
\NewDocumentCommand{\mynote}{+O{}+m}{%
  \begingroup
  \tcbset{%
    noteshift/.store in=\mynote@shift,
    noteshift=1.5cm
  }
   \begin{tcolorbox}[nobeforeafter,
    enhanced,
    sharp corners,
    toprule=1pt,
    bottomrule=1pt,
    leftrule=0pt,
    rightrule=0pt,
    colback=yellow!20,
    #1,
    left skip=\mynote@shift,
    right skip=\mynote@shift,
    overlay={\node[right] (mynotenode) at ([xshift=-\mynote@shift]frame.west) {\textbf{Note:}} ;},
    ]
    #2
  \end{tcolorbox}
  \endgroup
  }

\makeatother
\usetikzlibrary{decorations.pathmorphing,shapes}

\usetikzlibrary{arrows,automata,positioning} %

\usetikzlibrary{positioning,arrows.meta,decorations.pathmorphing,calc,shapes.multipart,automata,matrix}

\tikzstyle{pta}=[auto, ->, >=stealth']
\tikzstyle{every node}=[initial text=]
\tikzstyle{location}=[rectangle, rounded corners, minimum size=12pt, draw=black, fill=blue!10, inner sep=3.5pt]
\tikzstyle{invariant}=[draw=black, dotted, fill=yellow, inner sep=1pt, node distance=0] %
\tikzstyle{locguard}=[fill=cyan!30, inner sep=1pt, node distance=0] %
\tikzstyle{final}=[double, fill=blue!50]

\tikzstyle{urgent}=[fill=yellow, thick, dotted] %
\tikzstyle{private}=[fill=red,thick]

\newcounter{sarrow}

\iffinal
\usepackage[disable]{todonotes}
\else
\usepackage{todonotes}
\fi

\newcommand{\sk}[1]{}

\iffinal

\newcommand{\remove}[1]{}
\else
\usepackage{soul}

\newcommand{\remove}[1]{{\color{red}\st{#1}}}

\fi

\title{Parameterized Verification of\\ Disjunctive Timed Networks\thanks{%
	\ifdefined\VersionForArXiV
		This is the author (and extended) version of the manuscript of the same name published in the proceedings of the 25th International Conference on Verification, Model Checking, and Abstract Interpretation (\href{https://popl24.sigplan.org/home/VMCAI-2024}{VMCAI~2024}).
		The final authenticated version is available online at:
			\href{https://doi.org/10.1007/978-3-031-50524-9_6}{\nolinkurl{https://doi.org/10.1007/978-3-031-50524-9_6}}.
	\fi{}
	This work is partially supported by ANR-NRF ProMiS (ANR-19-CE25-0015 / 2019 ANR NRF 0092) and by ANR BisoUS (ANR-22-CE48-0012).
}}

\titlerunning{Parameterized Verification of Disjunctive Timed Networks}

\author{
	\'Etienne Andr\'e\inst{1}
		\homepage{https://lipn.univ-paris13.fr/~andre/}%
		\orcidID{https://orcid.org/0000-0001-8473-9555}
		\and
	Paul Eichler\inst{2}
	\orcidID{https://orcid.org/0009-0008-6117-318X}
	\and
	Swen Jacobs\inst{2}
		\homepage{https://swenjacobs.github.io/}%
		\orcidID{https://orcid.org/0000-0002-9051-4050}
		\and
	Shyam Lal Karra\inst{2}
	\orcidID{https://orcid.org/0009-0000-6859-4106}
}

\authorrunning{É. André, P. Eichler, S. Jacobs and S. Karra} 

\institute{%
	Université Sorbonne Paris Nord, LIPN, CNRS UMR 7030, \\ F-93430 Villetaneuse, France
	\and
	CISPA Helmholtz Center for Information Security,\\ Germany%
}

\usepackage[firstpage]{draftwatermark}  %
\SetWatermarkAngle{0}
\SetWatermarkText{\raisebox{12.5cm}{%
  \hspace{0.7cm}%
\href{https://doi.org/10.5281/zenodo.8337446}{\includegraphics[scale=0.09]{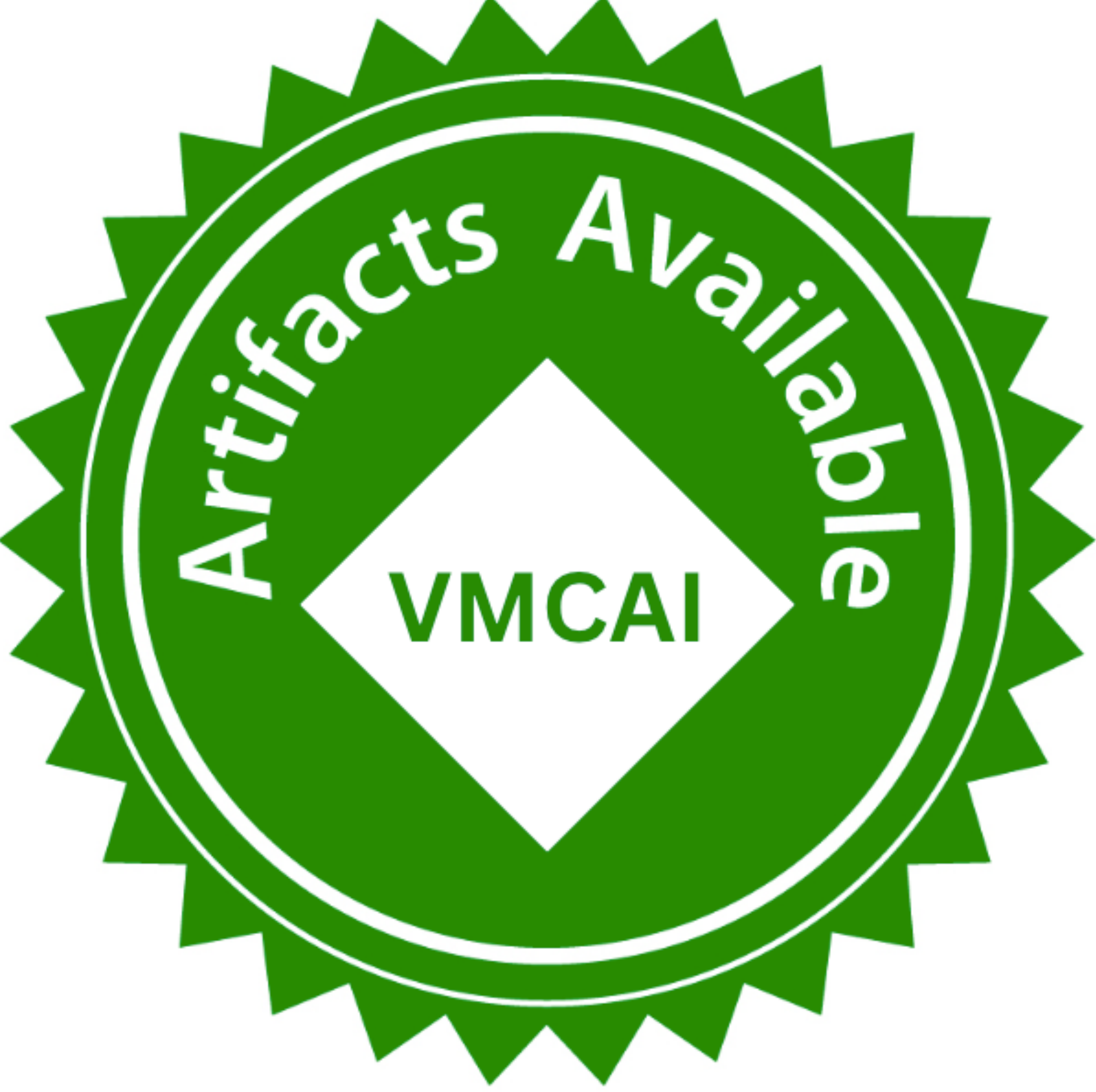}}%
  \hspace{8cm}%
  \includegraphics[scale=0.09]{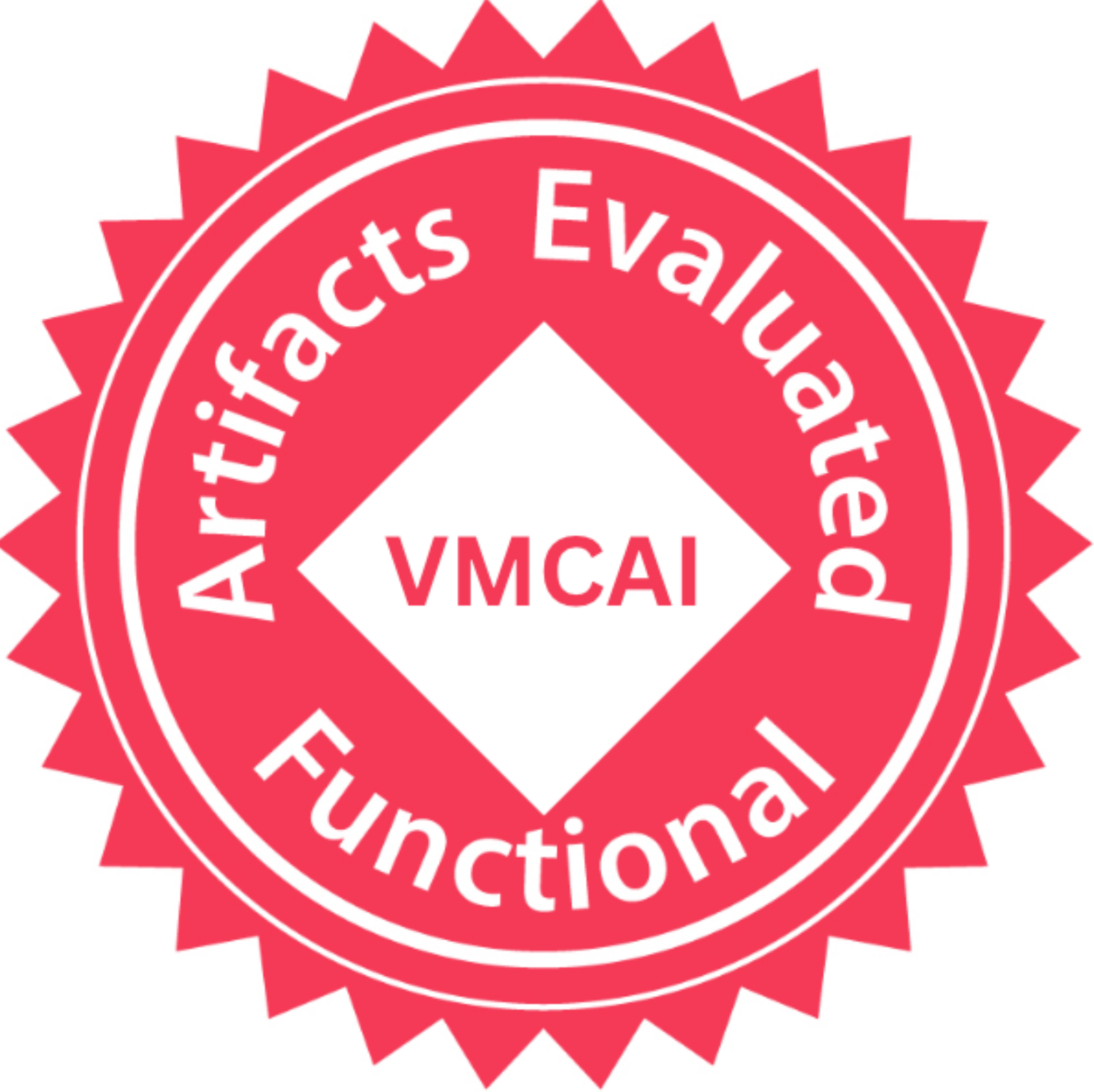}%
}}

\begin{document}
	
	\maketitle

	\begin{abstract}
		We introduce new techniques for the parameterized verification of disjunctive timed networks (\DTN{}s), i.e., networks of timed automata (TAs) that communicate via \emph{location guards} that enable a transition only if there is another process in a given location.
		This computational model has been considered in the literature before, example applications are gossiping clock synchronization protocols or planning problems.
		We address the minimum-time reachability problem (\minreach) in \DTN{}s, and show how to efficiently solve it based on a novel zone-graph algorithm. We further show that solving \minreach allows us to construct a “summary” TA capturing exactly the possible behaviors of a single TA within a \DTN{} of arbitrary size. The combination of these two results enables the parameterized verification of \DTN{}s, while avoiding the construction of an exponential-size cutoff system required by existing results.
		Additionally, we develop sufficient conditions for solving \minreach and parameterized verification problems even in certain cases where locations that appear in location guards can have clock invariants, a case that has usually been excluded in the literature.
		Our techniques are also implemented, and experiments show their practicality.

	\end{abstract}

	\section{Introduction}
\label{sec:intro}

Many computer systems today are distributed and rely on some form of synchronization between largely independent processes to reach a common goal. 
Formally reasoning about the correctness of such systems is difficult, since correctness guarantees are expected to hold regardless of the number of processes ---a problem that is also known as \emph{parameterized verification}~\cite{AD16,DBLP:reference/mc/AbdullaST18}, since the number of processes is considered as a parameter of the system.
Parameterized verification is undecidable even in very restricted settings, \eg{} for safety properties of systems composed of finite-state processes with rather limited communication, such as token-passing or transition guards~\cite{DBLP:journals/ipl/Suzuki88,DBLP:conf/cade/EmersonK00}.
However, many classes of systems and properties have been identified that have decidable parameterized verification problems~\cite{DBLP:series/synthesis/2015Bloem,DBLP:journals/dc/AminofKRSV18,DBLP:journals/dc/EsparzaJRW21}, usually with finite-state processes.

Systems and properties that involve timing constraints, such as clock synchronization algorithms or planning of time-critical tasks, cannot be adequately modeled with finite-state processes.
A natural model for such processes are timed automata (TAs)~\cite{DBLP:journals/tcs/AlurD94}, and the parameterized verification of systems of TAs, also called \emph{timed networks}, has already received some attention in the literature.
For models with powerful communication primitives, such as $k$-process synchronization or broadcast, safety properties are decidable if every process has a single clock, but they are undecidable if processes have multiple clocks, and liveness properties are undecidable regardless of the number of clocks~\cite{DBLP:journals/tcs/AbdullaJ03,DBLP:conf/lics/AbdullaDM04,ADRST16}.

In \emph{disjunctive timed networks} (\DTN{}s), communication is restricted to disjunctive guards that enable a transition only if at least one process is in a given location.
{This a commonly met communication primitive in the literature, having been studied for finite-state systems under the notion of guarded protocols~\cite{DBLP:conf/cade/EmersonK00,DBLP:conf/lics/EmersonK03,DBLP:conf/vmcai/AusserlechnerJK16,DBLP:conf/vmcai/JacobsS18}, and more recently in immediate observation Petri nets~\cite{DBLP:conf/apn/EsparzaRW19,DBLP:conf/concur/RaskinWE20} and immediate observation population protocols~\cite{DBLP:conf/concur/EsparzaGMW18}.}
For \DTN{}s, decidability is obtained also for multiple clocks per process and for liveness properties~\cite{SpalazziS20}. 

However, the existing results on timed networks have no or very limited support for \emph{location invariants} (which can force an action to happen before a time limit runs out), and the decidability result for \DTN{}s relies on the construction of a product system that can be hard to check in practice.
Moreover, to the best of our knowledge, no techniques exist for the computation of real-time properties such as minimum-time reachability in timed networks (for single TAs, see, \eg{}~\cite{ATP04,BDR08,ZNL16}).

In this paper, we show that minimum-time reachability can be effectively computed in \DTN{}s, which also leads to a more efficient parameterized verification technique for safety and liveness properties, and we provide conditions under which location invariants can be supported.
\DTN{}s are an interesting computational model, since (even without clock invariants) they can express classes of models where a resource, encoded as a location, is produced or unlocked once for all; this is the case of the whole class of planning problems, and some problems in chemical reaction networks.
Moreover, disjunctive networks have been used to verify security problems in grid computing~\cite{DBLP:journals/fgcs/PagliarecciSS13}, and real-time properties are also of interest in such applications~\cite{MBS11}.
Another natural application area are clock synchronization protocols, which we consider in the following.

\smartpar{Motivating Example: Gossiping Clock Synchronization}
Consider the example in \cref{fig:gtpProtocol}, depicting a simple clock synchronization protocol.
The semantics of a single process is the same as for standard TAs.
Invariants are depicted using dotted yellow boxes; $\clock \assign 0$ denotes reset of clock x.
As a synchronization mechanism, some transitions are \emph{guarded} by a location~$h_i$ (such location guards are highlighted in light blue), \ie{} they may only be taken by one process if another process is in~$h_i$.
Arrows marked with $h_0,h_1$ stand for two transitions, each guarded by one of the~$h_i$.

Processes are synchronized when they are in a location $h_i$ for $i \in \{0,1\}$, and they should move to $h_{(i+1 \mod 2)}$ after two time units. 
However, they can non-deterministically lose synchronization and move to a location $\ell_i$, where they move according to different timing constraints.
Via $\ell_{sy}$, they can re-synchronize with processes that are still in a location~$h_i$.

A version of this example was considered by Spalazzi and Spegni~\cite{SpalazziS20}. 
Notably, their version came \emph{without} clock invariants on locations $h_i$, \ie{} processes can stay in these locations forever, which is clearly not the intended behavior.
Besides not supporting location invariants, their parameterized verification\LongVersion{ technique} requires the construction of a product system, which is fine for this small example but quickly becomes impractical for bigger examples.

\begin{figure}[tb]
  \scalebox{1}{
			    \begin{tikzpicture}[font=\footnotesize]
        \begin{scope}[rounded corners]
         \node [location] (q2) at (3,0) { $h_0$ };
         \node [location] (q3) at (7,0) { $h_1$ };
          \node [location] (q4) at (0, -2.5) { $\ell_{sy}$ };
         \node [location] (q5) at (3,-2.5) { $\ell_0$ };
         \node [location] (q6) at (7,-2.5) { $\ell_1$ };
         
         \node [ invariant,above =of q2]  {$\clock \leq 2$};
         \node [ invariant,above =of q3]  {$\clock \leq 2$};
         \node [ invariant,below =of q5]  {$\clock \leq 4$};
         \node [ invariant,below =of q6]  {$\clock \leq 4$};
         \node (fake) at (2,0){};
        \end{scope}
         \begin{scope}[->, >=stealth, every node/.style={scale=0.9}]

         \draw (fake)--(q2);

        \path[] (q4) edge[]    (q5);
         \path[] (q5) edge[bend left]    (q4);

         \path[] (q2) edge[]  node[above]{$\clock=2,\clock \assign 0$}  (q3);
          \path[] (q4) edge[sloped]  node[above]{$\clock\assign 0$}node[below, locguard]{$h_0$}  (q2);
           \path[] (q4) edge[sloped]  node[above, pos = 0.25]{$\clock \assign 0$}node[below, locguard, pos=0.25]{$h_1$}  (q3);
            \path[] (q2) edge[sloped]  node[above, pos=0.75]{$\clock \neq 2, \clock \assign 0$} node[below,locguard, pos=0.75]{$h_0,h_1$}  (q6);
            \path[] (q3) edge[sloped]  node[above, pos=0.75]{$\clock \neq 2, \clock \assign 0$} node[below, locguard,pos=0.75]{$h_0,h_1$}  (q5);
             \path[] (q4) edge[]    (q5);
             \path[] (q5) edge[]  node[above]{$\clock \geq 1, \clock \assign 0$}   (q6);
             \path[] (q6) edge[bend left]  node[below]{$\clock \geq 1,\clock \assign 0$}   (q5);
             \path[] (q3) edge[bend right]  node[above]{$\clock = 2, \clock \assign 0$}   (q2);
        \end{scope}
        \end{tikzpicture}
 }
\caption{TA with disjunctive guards for gossiping clock synchronization}%
  \label{fig:gtpProtocol}
  \end{figure}

\smartpar{Contributions}
In this paper, we provide novel decidability results for \DTN{}s:
\begin{enumerate}
    \setlength{\itemsep}{0pt}
    \setlength{\parskip}{0pt}
    \setlength{\parsep}{0pt}
    \item We state the \emph{minimum-time reachability problem} (\minreach) for \DTN{}s, \ie{} computing the minimal time needed to reach any location of a TA, in networks with an arbitrary number of processes.
		We develop a technique to efficiently solve the problem for TAs without location invariants, based on a specialized zone-graph computation.

		\item We show that solving \minreach allows us to construct a \emph{summary automaton} that captures the local semantics of a process, \ie{} has the same set of possible executions as any single process in the \DTN. This allows us to decide parameterized verification problems without resorting to the global semantics in a product system that grows exponentially with the cutoff. 
		\item We show that under certain conditions, parameterized verification is still decidable even in the presence of location invariants for 1-clock TAs.
		We develop new domain-specific proof techniques that allow us to compute summary automata and cutoffs for these cases.
    \item We experimentally evaluate our techniques on variants of the clock synchronization case study and a hand-crafted example, demonstrating its practicality and its benefits compared to the cutoff-based approach. 
\end{enumerate}

\smartpar{Related work}
Infinite-state systems have been of high interest to the verification community. 
Their infiniteness can stem from two causes. 
First,  the state space of a single process can be infinite, as in TAs~\cite{DBLP:journals/tcs/AlurD94}. 
In order to obtain decidability results (for single TAs), abstraction techniques have been proposed in the literature (\eg{}~\cite{BehrmannBLP06,HSW13,BGHSS22}).
Second, infinite-state systems (effectively) arise if we consider parameterized systems that consist of an arbitrary number of components.
Apt and Kozen~\cite{DBLP:journals/ipl/AptK86} showed that reachability is undecidable in systems composed of an arbitrary number of finite-state systems, and Suzuki~\cite{DBLP:journals/ipl/Suzuki88} showed undecidability even for uniform finite-state components with very limited means of communication.
However, the parameterized verification problem remains decidable for many interesting classes of systems with finite-state processes~\cite{DBLP:conf/cade/EmersonK00,DBLP:journals/ijfcs/EmersonN03,DBLP:series/synthesis/2015Bloem,DBLP:conf/concur/ClarkeTTV04,DBLP:journals/fmsd/BouajjaniHV08,DBLP:conf/icfem/HannaSBR10,DBLP:journals/dc/AminofKRSV18}.

According to this classification, the parameterized verification of networks of TAs deals with infiniteness in two dimensions~\cite{DBLP:journals/tcs/AbdullaJ03,DBLP:conf/lics/AbdullaDM04}, and similarly the verification of timed Petri nets~\cite{DBLP:journals/tcs/JonesLL77}.
Here, the notion of urgency (which is closely related to location invariants) makes the reachability problem undecidable, even for the case where each TA is restricted to a single clock.
Our setting also shares similarities with timed \emph{ad hoc} networks~\cite{ADRST16}, where processes are TAs that communicate by broadcast in a given topology.
However, even simple problems like parameterized reachability quickly become undecidable~\cite{ADRST16}, \eg{} when processes contain two clocks if the topology is a clique (\ie{} all nodes can communicate with each other), or even when the nodes are equipped with a single clock if the topology is a star of diameter five.
In our setting, the topology is fixed to a clique, and communication is not based on actions, but on location guards.
In addition, \cite{DBLP:journals/tcs/AbdullaJ03,ADRST16} can only decide reachability properties and %
do not seem to be usable to derive minimum-time reachability results.

As the closest work to ours, Spalazzi and Spegni~\cite{SpalazziS20} consider networks of TAs with disjunctive guards, and use the cutoff method to obtain decidability results.
They show that, while cutoffs do not always exist, they do exist for a subclass where the underlying TAs do not have invariants on locations that appear in guards.
In our work, we show how to decide parameterized verification problems by constructing a summary automaton instead of an exponentially larger cutoff system, and we present new decidability results that address the limitation on location invariants.

A similar idea to our summary automaton for parameterized verification of systems communicating via pairwise rendez-vous has also appeared in \cite{DBLP:journals/dc/AminofKRSV18}, but in their case this is a Büchi automaton, whereas we need to use a TA. %
\LongVersion{%
Finally, \cite{MBS11} defines a timed calculus for timed systems, sharing similarities with~\cite{DBLP:conf/formats/AbdullaDRST11}.
}

Parameterized verification \emph{with timing parameters} was considered in~\cite{ADFL19}, with mostly undecidability results\LongVersion{, except for restricted subclasses}; in addition, the few decidability results only concern \emph{emptiness} of parameter valuations sets, and not synthesis, and therefore these parametric results cannot be used to encode a minimum-time reachability using a parameter minimization.

\smartpar{Outline}
\cref{section:preliminaries} recalls TAs and \DTN{}s, and states the parameterized verification problem.
\cref{section:algorithm:DTN-} introduces a technique for efficient parameterized verification of \DTN{}s where locations that appear in location guards cannot have invariants, while \cref{sec:cutoffs} proposes another technique that supports such invariants under certain conditions.
\cref{section:evaluation} shows the practical evaluation of our algorithms on several examples, and
\cref{section:conclusion} concludes the paper. %

	\section{Preliminaries}\label{section:preliminaries}
We define here TAs with location guards (also: guarded TAs) and networks of TAs\LongVersion{(\cref{sec:model})}, followed by parameterized verification problems and cutoffs (\cref{sec:paramver}), and finally the standard symbolic semantics of TAs (\cref{sec:symbolic}).

\LongVersion{\subsection{System Model}}\label{sec:model}
Let $\clocks$ be a set of clock variables. A \emph{clock valuation} is a mapping $v: \clocks \rightarrow \mQz$.
We denote by~$\mathbf{0}$ the clock valuation that assigns $0$ to every clock, and by $v+\delta$ for $\delta \in \mQz$ the valuation such that $(v+\delta)(c)=v(c)+\delta$ for all $c \in \clocks$.
We call \emph{clock constraints} $\clockcons$ the terms of the following grammar:%

{\centering

$\ccons \Coloneqq \top \mid \ccons \wedge \ccons \mid c \sim c' + \intconstant \mid c \sim \intconstant \text{ with } \intconstant \in \mathbb{N}, \, c, c' \in \clocks, \, {\sim} \in \{<, \leq, =, \geq, > \}.$

}

A clock valuation $v$ is said to \emph{satisfy} a clock constraint $\ccons$, written as $v \models \ccons$, if $\ccons$ evaluates to true after replacing every $c \in \clocks$ with its value $v(c)$.%

\smartpar{Guarded Timed Automaton (\gta)}
A \emph{\gta} $A$ is a tuple $(\Loc,\hat{\loc},\clocks,\mathcal{T},\invariant)$:
\begin{itemize}
    \setlength{\itemsep}{0pt}
    \setlength{\parskip}{0pt}
    \setlength{\parsep}{0pt}
    \item $\Loc$ is a finite set of locations with \emph{initial location} $\hat{\loc}$,
    \item $\clocks$ is a finite set of clock variables,
    \item $\mathcal{T} \subseteq \Loc \times \clockcons \times 2^\clocks \times  \left( \Loc \cup \{\top\} \right) \times \Loc$ is a transition relation, and
    \item $\invariant: \Loc \rightarrow \clockcons$ assigns to every location $\loc$ an \emph{invariant} $\invariant(\loc)$.
\end{itemize}

Intuitively, a transition $\tau = (\loc,g,r,\gamma,\loc') \in \mathcal{T}$ takes the automaton from location $\loc$ to~$\loc'$, it can only be taken if \emph{clock guard} $g$ and \emph{location guard} $\gamma$ are both satisfied, and it resets all clocks in~$r$.
We also write $\sguard{\tau}$ for~$\gamma$.
Note that satisfaction of location guards is only meaningful in a \emph{network} of~TAs, formally defined subsequently.
Intuitively, a location guard is satisfied if it is $\top$ or if another automaton in the network currently occupies location~$\gamma$.
We say that $\gamma$ is \emph{trivial} if $\gamma = \top$, and we write $\LocGuards(A) \subseteq \Loc$ for the set of non-trivial location guards that appear in~$A$.\footnote{%
    Note that since $\LocGuards(A)$ is effectively a set of locations, we also use it as such.}
We call $\loc \in \LocGuards(A)$ also a \emph{guard location}.
We say that a location $\loc$ \emph{does not have an invariant} if $\invariant(\loc)=\top$.

\begin{example}
    
    \cref{figure:example-gTA} shows an example gTA with 5~locations and one clock~$\clock$.
		Location $\hat{\loc}$ is the initial location, with a transition to~$\loc_0$ that has a clock constraint $\clock=2$, and a transition to~$\loc_2$ that resets~$\clock$ (which in figures is denoted as $\clock \leftarrow 0$ for readability) and a location guard $\loc_1$.
		Location $\loc_0$ has an invariant %
		$x \leq 4$.
		The transition from~$\loc_0$ to~$\hat{\loc}$ has location guard~``$\loc_0$'', requiring another process to be in location~$\loc_0$ to take this transition.
\end{example}

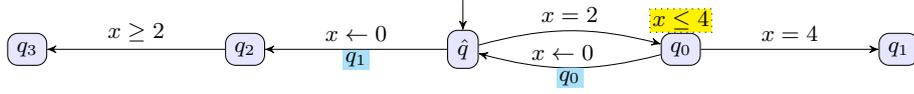
\begin{figure}[tb]
	\centering
	\resizebox{\linewidth}{!}{
	\begin{tikzpicture}[pta, bend angle=15, font=\footnotesize]

		\node[location] at (0, 0)  (q0) {$\hat{\loc}$};
		\node[location] at (3, 0) (q1) {$\loc_0$};
		\node[location] at (6, 0) (q2) {$\loc_1$};
		\node[location] at (-3, 0) (q3) {$\loc_2$};
    \node[location] at (-6,0) (q4) {$\loc_3$};

		\node[invariant, above=of q1] {$\clock \leq 4$};

		\node (fake) at (0,0.8){};
	
		\path (q0) edge[bend left] node[]{$\clock=2$} (q1);

		\path (q1) edge[bend left] node[above,xshift=-.3em]{$\clock \assign 0$} node[below,locguard]{$\loc_0$} (q0);
		\path (q1) edge[] node[above]{$\clock = 4$} (q2);

    \path (q0) edge[] node [above]{$\clock \assign 0$} node[below,locguard]{$q_1$}(q3);
    \draw (fake)--(q0);

    \path (q3) edge[] node[above]{$\clock \geq 2$} (q4);

	\end{tikzpicture}}
	\caption{A gTA example}
	\label{figure:example-gTA}
\end{figure}

Given a gTA~$A$, we denote by $\UG{A}$ the unguarded version of~$A$ \ie{} ignoring all location guards (or equivalently, replacing each location guard with $\top$).
Note that $\UG{A}$ is a~TA according to standard definitions in the literature~\cite{DBLP:journals/tcs/AlurD94}.

A \emph{configuration} of a gTA $A$ is a pair $(\loc,v)$, where $\loc \in \Loc$ and $v: \clocks \rightarrow \mQz$ is a \emph{clock valuation}.
When considering a \gta{}~$A$ in isolation, its semantics is the usual semantics of the TA~$\UG{A}$, \ie{} ignoring location guards.
That is, a \emph{delay transition} is of the form $(\loc,v) \rightarrow^\delta (\loc,v+\delta)$ for some {$\delta \in \mQz $} such that $\forall \delta' \in [0, \delta] : v+\delta' \models \invariant(\loc)$.
A \emph{discrete transition} is of the form $(\loc,v) \rightarrow^\tau (\loc',v')$, where $\tau=(\loc,g,r,\gamma,\loc') \in \mathcal{T}$, $v \models g$,
$v'[c]=0$ if $c \in r$ and $v[c]$ otherwise,
and $v' \models \invariant(\loc')$.
We say that the transition is \emph{based on $\tau$}.
We write $(\loc,v) \rightarrow (\loc',v')$ if either $(\loc,v) \rightarrow^\delta (\loc',v')$ or $(\loc,v) \rightarrow^\tau (\loc',v')$.
For convenience, we assume that for every location $\loc$ there is a discrete \emph{stuttering transition} $\epsilon = (\loc,\top,\emptyset,\top,\loc)$ that can always be taken and does not change the configuration.

We write $(\loc,v)\xrightarrow{\delta,\tau}(\loc',v')$ if there is a delay transition $(\loc,v) \rightarrow^\delta (\loc,v+\delta)$ followed by a discrete transition $(\loc,v+\delta) \rightarrow^\tau (\loc',v')$.
Then, a \emph{timed path} of~$A$ is a (finite or infinite) sequence\footnote{Wlog, we assume that the first transition of a timed path is a delay transition and that delay transitions and discrete transitions alternate.} $\lcomputation = (\loc_0,v_0) \xrightarrow{\delta_0,\tau_0} (\loc_0,v_0) \xrightarrow{\delta_1,\tau_1}  \ldots$.

For a finite timed path $\lcomputation = (\loc_0,v_0) \xrightarrow{\delta_0,\tau_0} \ldots \xrightarrow{\delta_{l-1},\tau_{l-1}} (\loc_l,v_l)$, let $\finalq(\lcomputation)=\loc_l$ be the final location of~$\lcomputation$, $\totaltime(\lcomputation) = \sum_{0 \leq i < l} \delta_i$ the total time delay of~$\lcomputation$.
We write $(\loc,v) \rightarrow^* (\loc',v')$ if there is a (finite) timed path $\lcomputation = (\loc_0,v_0) \xrightarrow{\delta_0,\tau_0} \ldots \xrightarrow{\delta_{l-1},\tau_{l-1}} (\loc_l,v_l)$.
A timed path is a \emph{computation} if $\loc_0=\hat{\loc}$ and $v_0=\mathbf{0}$.
The \emph{language} of~$A$, denoted $\Lg(A)$, is the set of all of its computations.

\begin{example}
Consider $\UG{A}$ for the gTA $A$ in \cref{figure:example-gTA}.
From its initial configuration $(\hat{\loc},\mathbf{0})$, there can be arbitrary delay transitions (since $\hat{\loc}$ does not have an invariant), and after a delay of $\delta \geq 0$ we can take a transition to $\loc_2$ which resets the clock $x$, \ie{} we arrive at configuration $(\loc_2,\mathbf{0})$. The location guard on this transition is ignored, since we consider $\UG{A}$.
In contrast, the transition from $\hat{\loc}$ to $\loc_0$ has a clock constraint that needs to be observed, \ie{} we can only take the transition after we reach a configuration $(\hat{\loc},v)$ with $v(x)=2$.
Then, the gTA can only stay in~$\loc_0$ as long as the invariant $x \leq 4$ is satisfied.
\end{example}

\smartpar{Network of TAs} For a given \gta $A$, we denote by $A^n$ the parallel composition $A \parallel \cdots \parallel A$ of $n$ copies of~$A$, also called a \textit{network of TAs} ($\mathrm{NTA}$ for short). A copy of~$A$ in the NTA will also be called a \emph{process}.

A \emph{configuration} $\nConfig$ of an NTA %
$A^n$ is a tuple
$\big((\loc_1,v_1), \ldots, (\loc_n,v_n)\big)$, where every $(\loc_i,v_i)$ is a configuration of~$A$.
The semantics of~$A^n$ can be defined as a \emph{timed transition system} $(\nConfigSet, \nConfigInit,T)$, where $\nConfigSet$ denotes the set of all  configurations of~$A^n$,  $\nConfigInit$ is the unique initial configuration $(\hat{\loc},\mathbf{0})^n$, and the transition relation $T$ is the union of the following delay and discrete transitions:

\begin{description}
    \setlength{\itemsep}{0pt}
    \setlength{\parskip}{0pt}
    \setlength{\parsep}{0pt}
    \item[delay transition] $ \transition {\big ((\loc_1,v_1),\ldots,(\loc_n,v_n) \big)} {\delta} {\big ((\loc_1,v_1+\delta),\ldots,(\loc_n,v_n+\delta)\big)}$\\
    if $\forall i \in [1,n] : \forall \delta' \leq \delta : v_i+\delta' \models \invariant(\loc_i)$, \ie{} we can delay $\delta \in \mQz$ units of time if all invariants $\invariant(\loc_i)$ will be satisfied until the end of the delay;
    \item[discrete transition] $\transition{\big((\loc_1,v_1),\ldots,(\loc_n,v_n)\big)} {(i, \tau)} {\big((\loc_1',v_1'),\ldots,(\loc_n',v_n')\big)}$ if
        \begin{enumerate}%
            \item $\transition{(\loc_i,v_i)} {\tau} {(\loc_i',v_i')}$ is a discrete transition of~$A$ with $\tau = (\loc_i,g,r,\gamma,\loc_i')$,
            \item $\gamma=\top$ or $\loc_j=\gamma$ for some $j \in [1,n]\setminus \{i\}$, and
            \item $\loc_j' = \loc_j$ and $v_j'=v_j$ for all $j \in [1,n]\setminus \{i\}$.
        \end{enumerate}%

\end{description}

That is, location guards $\gamma$ are interpreted as disjunctive guards: unless $\gamma=\top$, at least one of the other processes needs to occupy the location $\gamma$.

{We write $\nConfig \xrightarrow{\delta, (i,\tau)} \nConfig''$ for a delay transition $\nConfig \xrightarrow{\delta} \nConfig'$ followed by a discrete transition $\nConfig' \xrightarrow{(i,\tau)} \nConfig''$.
Then, a \emph{timed path} of~$A^n$ is a (finite or infinite) sequence $\pi = \nConfig_0 \xrightarrow{\delta_0, (i_0, \tau_0)} \nConfig_1 \xrightarrow{\delta_1, (i_1, \tau_1)} \cdots$.}
For a finite timed path $\pi = \nConfig_0 \xrightarrow{\delta_0,\nDisTrans{0}} \nConfig_1 \xrightarrow{\delta_1, \nDisTrans{1}} \cdots \xrightarrow{\delta_{l-1}, \nDisTrans{l-1}} \nConfig_l$, let $\finalc(\pi)= \nConfig_l$ be the final configuration of~$\pi$ and $\totaltime(\pi) = \sum_{0 \leq i < l} \delta_i$ the total time delay of~$\pi$. %
A timed path $\pi$ of~$A^n$ is a \emph{computation} if $\nConfig_0 = \nConfigInit$.
The \emph{language} of~$A^n$, denoted $\Lg(A^n)$, is the set of all of its computations.

We will also use \emph{projections} of these global objects onto subsets of the processes.
That is, if $\nConfig=\big((\loc_1,v_1),\ldots,(\loc_n,v_n)\big)$ and $\intervalPr = \{i_1, \ldots, i_k\} \subseteq [1,n]$, then $\nConfig|_{\intervalPr}$ is the tuple $\big((\loc_{i_1},v_{i_1}),\ldots,(\loc_{i_k},v_{i_k})\big)$, and we extend this notation in the natural way to computations $\pi|_{\intervalPr}$ and the language $\Lg_{|\intervalPr}(A^n)$.\footnote{%
    In particular, in the projection $\pi|_{\intervalPr}$ of a timed path $\pi$, any discrete transition of a process $j \notin \intervalPr$ is replaced by a stuttering transition $\epsilon$ of one of the processes $i \in \intervalPr$, and $\Lg_{|\intervalPr}(A^n)$ contains the projection for every possible choice of $i$.}
Similarly, for timed paths $\pi_1$ of $A^{n_1}$ and $\pi_2$ of $A^{n_2}$ we denote by $\pi_1 \parallel \pi_2$ their (non-necessarily unique) \emph{composition} into a timed path of $A^{n_1+n_2}$.
We write $\loc \in \nConfig$ if $\nConfig=\big((\loc_1,v_1),\ldots,(\loc_n,v_n)\big)$ and $\loc = \loc_i$ for some $i \in [1,n]$, and similarly $(\loc,v) \in \nConfig$.

We say that a location $\loc$ is \emph{reachable} in~$A^n$ if there exists a reachable configuration $\nConfig$ s.t.\ $\loc \in \nConfig$.
We denote by $\ReachLoc{A}$ the set of reachable locations in~$\UG{A}$, and by $\ReachLoc{A^n}$ the set of reachable locations in~$A^n$.

If $\pi$ is a timed path and $d \in \mQz$, then by $\pi^{\leq d}$ 
we denote the maximal prefix of~$\pi$ with $\totaltime(\pi^{\leq d})\leq d$%
, and similarly for timed paths $\lcomputation^{\leq d}$ of a single gTA.%
\footnote{Note that $\pi^{\leq d}$ is infinite if $\pi$ is infinite and has a total time delay smaller than $d$.}

\begin{example}
Consider a network with $2$ processes that execute the gTA in \cref{figure:example-gTA}.
When both processes are in the initial configuration $(\hat{\loc},\mathbf{0})$, the transition to $\loc_2$ is disabled because of the location guard $\loc_1$.
However, after a delay of $\delta=2$, one of them can move to $\loc_0$, and after another delay of $\delta=2$ from $\loc_0$ to $\loc_1$.
As $\loc_1$ is now occupied, the process that stayed in $\hat{\loc}$ can now take the transition to $\loc_2$.
\end{example}

\smartpar{Disjunctive Timed Network} %
A given \gta $A$ induces a \emph{disjunctive timed network} (\DTN{})\footnote{%
    We reuse terminology and abbreviations from~\cite{SpalazziS20}.
} $A^\infty$, defined as the following family of NTAs:
    $A^\infty = \{A^n: n \in \mathbb{N}_{>0}\}$. 
We define $\Lg(A^\infty) = \bigcup_{n \in \Nats} \Lg(A^n)$ and, for $\intervalPr = [1,i]$ with $i \in \Nats$, let $\Lg_{|\intervalPr}(A^\infty) = \bigcup_{n \geq i} \Lg_{|\intervalPr}(A^n)$.

We say that a gTA $A$ \emph{has persistent guard locations} if every location $\loc \in \LocGuards(A)$ has $\invariant(\loc)= \top$.
We denote by \DTNpersistent{} %
the class of disjunctive timed networks where $A$ has persistent guard locations.
That is, in a \DTNpersistent{} a location can only appear in a location guard if it does not have an invariant.

\subsection{Parameterized Verification Problems}\label{sec:paramver}
In this work, we are mostly interested in determining $\Lg_{|\intervalPr}(A^\infty)$ for some fixed finite $\intervalPr$, since this allows us to solve many parameterized verification problems.
This includes any local safety or liveness properties of a single process (with $\intervalPr=[1,1]$), as well as mutual exclusion properties (with $\intervalPr=[1,2]$) and variants of such properties for larger~$\intervalPr$.
Note that, even though we are interested in the language of a small number of processes, these processes still interact with an arbitrary number of identical processes in the network.\footnote{%
    One notable exception, \ie{} a parameterized verification problem that cannot be answered based on $\Lg_{|\intervalPr}(A^\infty)$, is the (global) deadlock detection problem, not considered here, or similar problems of simultaneous behavior of all processes.
    }

\smartpar{Cutoffs}
We call a \emph{family of gTAs} a collection $\famTemplates$ of gTAs, expressed in some formalism. 
E.g., let $\famTemplates_\top$ be the collection of all gTAs with persistent guard locations
and, for any~$k$, let $\famTemplates_k$ be the collection of all gTAs with $\left| \LocGuards(A) \right| \leq k$.%

Then, a \emph{cutoff} for a family of {TAs} $\famTemplates$ and a number of processes $m \in \mathbb{N}$ %
is a number $c \in \mathbb{N}$ such that for every $A \in \famTemplates$:
\begin{center}
$\Lg_{|[1,m]}(A^\infty) = \Lg_{|[1,m]}(A^c)$
\end{center}

Note that our definition of cutoffs requires language equality (for a fixed number $m$ of processes) between $A^c$ and $A^\infty$, which implies that $c$ is also a cutoff under other definitions from the literature which only require equivalence of~$A^c$ and $A^\infty$ wrt.\ certain logic fragments that define subsets of the language.

In particular, we can immediately state the following generalization of an existing result, which gives a cutoff for gTAs with $k$ locations used in guards, and without invariants on these locations. 
Let $\famTemplates_{\top,k} = \famTemplates_\top \cap \famTemplates_k$.

\begin{theorem}[follows from~\cite{SpalazziS20}]
\label{thm:Spalazzi}
For any $m \in \mathbb{N}$, $m+k$ is a cutoff for $\famTemplates_{\top,k}$.
\end{theorem}

\subsection{Symbolic Semantics of Timed Automata}
\label{sec:symbolic}

We build on the standard zone-based symbolic semantics of TAs~\cite{BY03}. In the following, let $A = (\Loc,\hat{\loc},\clocks,\mathcal{T},\invariant)$ be a gTA, interpreted as a standard TA.

\LongVersion{\smartpar{Zones and Zone Operations}}
A \emph{zone} is a clock constraint $\ccons$ (as defined in \cref*{sec:model}), representing all clock valuations that satisfy $\ccons$.
In the following we will use the constraint notation and the set (of clock valuations) notation interchangeably. %
We denote by $\mathcal{Z}$ the set of all zones.

A \emph{symbolic configuration} of~$A$ is a pair $(\loc,\zone)$, where $\loc \in \Loc$ and $\zone$ is a zone over\LongVersion{ the clock variables}~$\clocks$.
For a given symbolic configuration $(\loc,\zone)$, let $\zone^\uparrow = \{u+\delta \mid u \in \zone, \delta \in \mQz\}$, and $(\loc,\zone)^\uparrow = \big(\loc,\zone^\uparrow \cap \invariant(\loc) \big)$.
For a zone $\zone$ %
and a clock constraint~$g$, let $\zone \wedge g= \{ v \mid v \in \zone$ and $ v \models g \}$.

\smartpar{Zone Graph}
The \emph{zone graph} $ZG(A)= (S,S_0,\mathit{Act},\Rightarrow)$ of~$A$ is a transition system where
\begin{ienumerate}
\item $S$ is the (infinite) set of nodes of the form $(\loc,\zone)$ where $\loc \in \Loc$ and $\zone$ is a zone;
\item The initial node is $S_0 = (\hat{\loc},\mathbf{0})^\uparrow$;
\item For any two nodes $(\loc,\zone)$ and $(\loc',\zone')$, there is a \emph{symbolic transition}, $(\loc,\zone) \xRightarrow{\tau} (\loc',\zone')$ if there exists a transition $\tau=(\loc,g,r,\gamma,\loc') \in \mathcal{T}$ such that $(\loc',\zone') = \big(\loc', \{ v' \in \mathbb{R}^X_{\geq 0} \mid \exists v \in \zone \colon (\loc,v) \rightarrow^{\tau} (\loc',v')\}\big)^\uparrow$ and $\zone' \neq \emptyset$.
\item $\Rightarrow$ is the union of all $\xRightarrow{\tau}$, and $\Rightarrow^*$ is the transitive closure of~$\Rightarrow$.
\end{ienumerate}
A \emph{symbolic path} in a zone graph is a (finite or infinite) sequence of symbolic transitions $\sympath = (\loc_0,\zone_0)\xRightarrow{\tau_0}(\loc_1,\zone_1) \xRightarrow{\tau_1} \cdots $. %

	\section{Minimum-Time Reachability Algorithm for \DTNpersistent{}}\label{section:algorithm:DTN-}

In this section, we show how we can solve the \emph{minimum-time reachability problem} (\minreach), \ie{} determine reachability of every $\loc \in \Loc$ in a \DTNpersistent{} and compute minimal reachability times $\nInfMinReachT{\loc}$ for every reachable location $\loc$.
Solving \minreach is essential since, in a DTN with $A \in \famTemplates_\top$, the minimal reachability times completely determine when each disjunctive guard can be satisfied. 
We will show that this also allows us to determine $\Lg|_{[1,1]}(A^\infty)$ by ``filtering'' $\Lg(\UG{A})$ with respect to the $\nInfMinReachT{\loc}$-values.

Formally, we define%
\[\minReachT{\loc} = \min\Big(\big\{ d {\in \mQz} \mid \exists \lcomputation \in \Lg(\UG{A}) \text{ s.t.\ } \lcomputation^{\leq d} \text{ is finite and } \loc = \finalq(\lcomputation^{\leq d}) \big\}\Big)\text{,}\]
\noindent{}\ie{} the minimal global time to reach $\loc$ in~$\UG{A}$%
, and for $i \in \Nats \cup \{\infty\}$
\[\niMinReachT{i}{\loc} = \min\Big(\big\{ d {\in \mQz} \mid \exists \gcomputation \in \Lg(A^i) \text{ s.t. } \gcomputation^{\leq d} \text{ is finite  and } \loc \in \finalc(\gcomputation^{\leq d}) \big\}\Big)\text{,}\]
\noindent{}\ie{} the minimal global time such that one process reaches $\loc$ in a \DTN{} of size~$i$ (if $i \in \Nats$), or a network of any size (if $i=\infty$).

Then, \minreach is the problem of determining $\nInfMinReachT{\loc}$ for every $\loc \in \Loc$.

\begin{example}
{
  For the gTA in \cref{figure:example-gTA} we have $\minReachT{\loc_3}=2$, as the transition from $\hat{\loc}$ to~$\loc_2$ is immediately enabled, and we then need to wait $2$~time units for the transition from $\loc_2$ to~$\loc_3$ to be enabled.
  But we have $\niMinReachT{2}{\loc_3}=\nInfMinReachT{\loc_3}=6$, as one process needs to move to $\loc_1$ before another can take the transition to~$\loc_2$.
  Also note that in a network with one process, $\loc_3$ cannot be reached \ie{} $\niMinReachT{1}{\loc_3} = \infty$.
}

\end{example}

\begin{remark}\label{remark:infimum}
  Note that the \emph{minimum} may not always be defined, as the smallest reachable time may be an \emph{infimum} (\eg{} of the form $\niMinReachT{i}{\loc} > \intconstant$ for some $\intconstant \in \Nats$).
  To ease the exposé, we assume that the minimum is always defined, but all our constructions can be extended to work for both cases of minimum and infimum.
\end{remark}

A ``naive'' way to solve \minreach would be based on \cref{thm:Spalazzi}, our generalization of existing cutoff results for DTNs.\footnote{Minimum-time reachability cannot be expressed in their specification language, but it can be shown that their proof constructions preserve minimum-time reachability.}
I.e., we can consider the cutoff system $A^{1+\card{\LocGuards(A)}}$ and determine $\niMinReachT{1+\card{\LocGuards(A)}}{\loc}$ with standard techniques for TAs~\cite{CY92}, but this approach does not scale well in the size of $\LocGuards(A)$, as each additional clock (in the product system) may result in an exponentially larger zone graph~\cite{HS10,HSW12}.
Instead, we will show how to solve \minreach more efficiently, based on a specialized zone-graph construction that works on a single copy of the given gTA, and determines minimal reachability times sequentially.

To this end, note that all reachable locations can be totally ordered based on the order in which they can be reached. 
Formally, define ${\reachorder} \subseteq \ReachLoc{A^\infty} \times \ReachLoc{A^\infty}$ as follows: we have $\loc \reachorder \loc'$ if either
\begin{ienumerate}%
  \item $\niMinReachT{\infty}{\loc} < \niMinReachT{\infty}{\loc'}$, or
  \item $\niMinReachT{\infty}{\loc} = \niMinReachT{\infty}{\loc'}$ and there exists a computation $\gcomputation \in \Lg(A^\infty)$ that reaches $\loc$ before it reaches $\loc'$
    (\ie{} $\loc$ is the last location of a prefix of~$\gcomputation$ that does not contain~$\loc'$).
\end{ienumerate}%
Now consider $\RGuards(A)=\LocGuards(A) \cap \ReachLoc{A^\infty}$, the \emph{reachable guard locations} of $A$.
Then the following lemma states that, if $\{\loc_1, \ldots, \loc_n\}$ are the reachable guard locations of a \DTNpersistent{} ordered by $\reachorder$, then for every $i \in \{1, \ldots, n \}$ there is a computation of $A^i$ such that each process $j$ with $j \leq i$ reaches~$\loc_j$ at minimal time, and can stay there forever.

\begin{lemma}\label{lemma:locationguards-ordered}%
  Let $A$ be a \DTNpersistent{}.
  Let $\RGuards(A) =\{ \loc_1, \cdots, \loc_n \}$ such that $\loc_i \reachorder \loc_{i+1}$ for $i = 1, \ldots, n-1$.
  Then for every~$i \in \{1, \dots, n\}$, there exists a computation $\gcomputation \in \Lg(A^i)$ such that,
  for each $j \in \{1, \dots, i \}$, the projection $\lcomputation_j = \pi|_{[j,j]}$
  reaches $\loc_j$ at time $\nInfMinReachT{\loc_j}$, and stays there forever.
\end{lemma}
\begin{proof}
  Fix $i \in \{1, \dots, n\}$.
  First consider $j = 1$.
      By definition of $\reachorder$, $\loc_1$ is a guard location that is reachable at the earliest time, and any guard location $\loc'$ with $\niMinReachT{\infty}{\loc'} = \niMinReachT{\infty}{\loc_1}$ can be reached on a computation that reaches $\loc_1$ first.
			Thus, there exists a computation $\lcomputation$ of~$\UG{A}$ reaching $\loc_1$ at time $\niMinReachT{\infty}{\loc_1}$ without passing any non-trivial location guard.
      Note that $\lcomputation$ is also a computation $\gcomputation$ of $A^1$ with the desired properties.
			
      Now, for any $1 < j \leq i$, we can assume that there is a computation $\gcomputation$ of $A^{j-1}$ where for each of the guard locations $\{ \loc_1, \dots, \loc_{j - 1} \}$ there is a process that reaches it in minimal time and stays there forever (which is possible because $A$ is a \DTNpersistent{}).
			By definition of $\niMinReachT{\infty}{\loc_j}$ and $\reachorder$, there exists a computation $\gcomputation_j$ of $A^m$, for some $m$, that reaches $\loc_j$ in minimal time, and the transitions along $\gcomputation_j$ only depend on guard locations in $\{ \loc_1, \dots, \loc_{j - 1} \}$.
			If $\lcomputation_j$ is the local computation of the process that reaches $\loc_j$ in $\gcomputation_j$, then $\gcomputation' = \gcomputation \parallel \lcomputation_j$ is a computation of $A^j$ with the desired properties. \qed
\end{proof}

The following lemma formalizes the connection between computations of~$A$ and~$A^\infty$, stating that a computation of~$A$ is also the projection of a computation of~$A^\infty$ if no transition with a location guard~$\loc$ is taken before $\loc$ can be reached.

\begin{restatable}{lemma}{deltagmin}
\label{lem:deltagmin}

Let $A \in \famTemplates_\top$. A computation $\lcomputation \in \UG{A}$ is in $\Lg|_{[1,1]}(A^\infty)$ iff for every prefix $\lcomputation'$ of $\lcomputation$ that ends with a discrete transition %
$(\loc,v) \rightarrow^\tau (\loc',v')$ with $\sguard{\tau}\neq \top$, we have $\nInfMinReachT{\sguard{\tau}} \leq \totaltime(\lcomputation')$.

\end{restatable}

\iffull
\begin{proof}
For the ``if'' direction, assume that every discrete transition $\tau$ in $\lcomputation$ happens at a time that is greater than $\nInfMinReachT{\sguard{\tau}}$. By \cref{lemma:locationguards-ordered} there exists a computation $\gcomputation$ of $A^{|\LocGuards(A)|}$ in which every $\loc \in \RGuards(A)$ is reached at time $\nInfMinReachT{\loc}$, and the process that reaches $\loc$ stays there forever.
Thus, we have $(\lcomputation \parallel \gcomputation) \in \Lg(A^\infty)$, and therefore $\lcomputation \in \Lg|_{[1,1]}(A^\infty)$.

The ``only if'' direction is simple: if $\lcomputation$ contains a discrete transition $(\loc ,v) \rightarrow^\tau (\loc ',v')$ that happens at a time strictly smaller than $\nInfMinReachT{\sguard{\tau}}$, then by definition no other local run can be in location $\sguard{\tau}$ at this time (in any computation of $A^\infty$), and therefore $\lcomputation$ cannot be in $\Lg|_{[1,1]}(A^\infty)$. \qed
\end{proof}

\fi

\begin{example}
Consider again the gTA in \cref{figure:example-gTA} (ignoring the location invariant of~$\loc_0$%
, as otherwise it is not a \DTNpersistent{}).
Considering the possible behaviors of a process in a \DTNpersistent{}, we can assume that transition $\hat{\loc} \rightarrow \loc_2$ is enabled whenever the global time is at least $\nInfMinReachT{\loc_1}$.
This is because we can assume that another process moves to that location on the minimum-time path (and stays there).
\end{example}

We will show subsequently that \cref{lem:deltagmin} allows us to solve \minreach by
\begin{enumerate}
\item considering a variant~$\Au$ of~$A$ with a global clock variable %
  never reset, and
\item applying a modified zone-graph
    algorithm on a \emph{single instance} of~$\Au$.
\end{enumerate}
Working on a single instance of~$\Au$ will lead to an exponential reduction of time and memory when compared to a naive method based on cutoff results.

\subsection{Symbolic Semantics for \DTNpersistent{}}

As \cref{lem:deltagmin} shows, to decide whether a path with location guards can be executed in a \DTNpersistent{}, it is important to keep track of global time.
The natural way to do this is by introducing an auxiliary clock $\gclock$ that is never reset~\cite{BehrmannF01,Al-BatainehR017}.
We capture this idea by proposing the following definition of a DTN-Zone Graph.

\smartpar{DTN-Zone Graph}
Given a gTA $A = (\Loc,\hat{\loc},\mathcal{C},\mathcal{T},\invariant)$, let $A' = (\Loc,\hat{\loc},\mathcal{C}\cup\{\gclock\},\mathcal{T},\invariant)$.
Then the \emph{DTN-zone graph} $\ZoneGraph^\infty(A')$ is a transition system where
\begin{itemize}
    \setlength{\itemsep}{0pt}
    \setlength{\parskip}{0pt}
    \setlength{\parsep}{0pt}
\item the set $S$ of nodes and the initial node $(\hat{\loc},\zone_0)$ are as in $\ZoneGraph(A')$.

\item For any two nodes $(\loc,\zone)$ and $(\loc',\zone')$, there is
	\begin{itemize}
		\item a \emph{guarded transition}
        $(\loc,\zone) \xRightarrow{\tau, \gamma} (\loc',\zone')$ if there is a symbolic transition $(\loc,\zone) \xRightarrow{\tau} (\loc',\zone')$ in $\ZoneGraph(A')$ and $\gamma=\top$;
		\item a \emph{guarded transition} $(\loc,\zone) \xRightarrow{\tau, \gamma} (\loc',\zone')$
		if there exists a $\tau \in \mathcal{T}$ such that $(\loc',\zone') = (\loc', \{ v' \in \mathbb{R}^X_{\geq 0} \mid \exists v \in \zone \colon v(\gclock) \geq \nInfMinReachT{\gamma} \land (\loc,v) \rightarrow^{\tau} (\loc',v')\})^\uparrow$, $\zone' \neq \emptyset$ and $\gamma\neq\top$. I.e., in this case we effectively add $\gclock \geq \nInfMinReachT{\gamma}$ to the clock constraint of $\tau$.
		\end{itemize}

\item $\xRightarrow{\gamma}$ is the union of all $\xRightarrow{\tau, \gamma}$, and $\xRightarrow{\gamma}^*$ is the transitive closure of $\xRightarrow{\gamma}$
\end{itemize}

Let $\vmin$ be a function that takes as input a zone $\zone$ and a clock $c$, and returns the lower bound of $c$ in $\zone$.
Then we want (a finite version of) $\zgDTN$ to satisfy the following properties:
\begin{itemize}
\item \emph{soundness with respect to \minreach}, \ie{} if  $(\hat{\loc},\zone_0) \xRightarrow{\gamma}^* (\loc',\zone')$
is such that $\vmin(\zone', \gclock) \leq \vmin(\zone'', \gclock)$ for all nodes $(\loc',\zone'')$ reachable in $\ZoneGraph^\infty(A')$, then there exist $n \in \Nats$ and a timed path $\gcomputation = \nConfigInit \xrightarrow{\delta_0,\nDisTrans{0}} \nConfig_1 \xrightarrow{\delta_1, \nDisTrans{1}} \ldots \xrightarrow{\delta_{l-1}, \nDisTrans{l-1}} \nConfig_l$ of $A^n$ with $(\loc',v') \in \nConfig_l$ such that $v'\in \zone'$ and $v'(\gclock)=\vmin(\zone', \gclock)$.

\item \emph{completeness with respect to \minreach}, \ie{} if $\loc'$ is reachable in $A^\infty$ then there exists $(\hat{\loc},\zone_0) \xRightarrow{\gamma}^* (\loc',\zone')$ with $\vmin(\zone', \gclock)=\nInfMinReachT{\loc'}$.
\end{itemize}

Note that for a gTA~$A$, and $A'$ defined as above, 
 the zone graph $\ZoneGraph(\UG{A'})$ is sound w.r.t.~\minreach when considering executions of a single copy of $\UG{A}$.
Moreover, it is known that completeness in this setting can be preserved under a time-bounded exploration with a bound $\timeBound$ such that $\timeBound >\minReachT{\loc}$ for every $\loc \in \Loc$.
For $\ZoneGraph^\infty(A')$ and executions of $A^n$ however, having $\timeBound$ only slightly larger than $\minReachT{\loc}$ for every $\loc \in \Loc$ may not be sufficient to preserve completeness, as the following example shows:

\begin{example}
  We have already seen that for the gTA in \cref{figure:example-gTA} we have $\nInfMinReachT{\loc_3}=6$, even though $\minReachT{\loc} \leq 4$ for all $\loc$.
  Thus, if we choose $\timeBound = 4$, a time-bounded exploration will not find a path to $\loc_3$ in $\ZoneGraph^\infty(A')$.
\end{example}

\smartpar{Bounding minimal reachability times}
In the following, we compute an upper-bound on the minimum-time reachability in a \DTNpersistent{}, which will allow us to perform a time-bounded exploration, thus rendering the zone graph finite.

Let $\delta_{max} = \max \{ \minReachT{\loc} \mid \loc \in \ReachLoc{A}\}$.
Our upper bound is defined as $\UpperBound(A) = \delta_{max} \cdot (\card{\LocGuards(A)} + 1)$, \ie{} the maximum over the locations of the minimum to reach that location, times the number of location guards plus one.

\begin{restatable}{lemma}{reachability}
\label{lem:reachability}
For any given gTA $A = (\Loc,\hat{\loc},\mathcal{C},\mathcal{T},\invariant)$, we have
\begin{enumerate}
  \item $\ReachLoc{A} \supseteq \ReachLoc{A^\infty}$, and
  \item for all $\loc \in \ReachLoc{A^\infty}$, $\nInfMinReachT{\loc} \leq \UpperBound(A)$.
\end{enumerate}
\end{restatable}

\iffull
\begin{proof}
Point~1 directly follows from the fact that $\Lg|_{[1,1]}(A^\infty) \subseteq \Lg(A)$.
For point~2, let $\ReachLoc{A^\infty} = \{ \loc_1, \ldots, \loc_k\}$ and assume wlog.\ $\loc_1 \reachorder \loc_2 \reachorder \ldots \reachorder \loc_k$.
Then the minimal timed path to reach $\loc_1$ must be such that it only uses discrete transitions with trivial location guards.
Therefore, we have $\nInfMinReachT{\loc_1}=\minReachT{\loc_1}\leq \delta_{max}$.
Inductively, we get for every $\loc_i$ with $i>1$ that it can only use transitions that rely on $\{\loc_1,\ldots,\loc_{i-1}\} \cap \LocGuards(A)$, and therefore $\nInfMinReachT{\loc_i}\leq (i-j) \cdot \delta_{max}$, where $j = \card{\{\loc_1,\ldots,\loc_{i-1}\} \setminus \LocGuards(A)}$.
\qed
\end{proof}

\fi

Thus, it is sufficient to perform a time-bounded analysis with $\UpperBound(A)$ as a time horizon.
Formally, let $\ZoneGraph_{\UpperBound(A)}^\infty(A')$ be the finite zone graph obtained from $\ZoneGraph^\infty(A')$ by intersecting every zone with $\gclock \leq \UpperBound(A)$.

\begin{restatable}{lemma}{minreachsingle}

\label{lem:minreach-single}
For $A = (\Loc,\hat{\loc},\mathcal{C},\mathcal{T},\invariant) \in \famTemplates_\top$, let $A' = (\Loc,\hat{\loc},\mathcal{C}\cup\{\gclock\},\mathcal{T},\invariant)$. Then $\bzgDTN$ is sound and complete with respect to \minreach.

\end{restatable}

\iffull
\begin{proof}
\emph{Soundness w.r.t.~\minreach}: let
$\sympath = (\hat{\loc},\zone_0) \xRightarrow{\tau_0,\gamma} (\loc_1,\zone_1) \xRightarrow{\tau_1,\gamma} \cdots \xRightarrow{\tau_{l-1},\gamma} (\loc_l,\zone_l)$ in $\bzgDTN$ with $\vmin(\zone_l,t) \leq \vmin(\zone,t)$ for all $(\loc_l,\zone)$ reachable in $\bzgDTN$.

We prove existence of a computation
$\gcomputation = \nConfigInit \xrightarrow{\delta_0,\nDisTrans{0}} \nConfig_1 \xrightarrow{\delta_1, \nDisTrans{1}} \ldots \xrightarrow{\delta_{l-1}, \nDisTrans{l-1}} \nConfig_l$ of $A^n$ (for some~$n$) with $v\in \zone_l$ and $v(\gclock)= \vmin(\zone_l,t)$
for some $(\loc,v) \in \nConfig_l$  %
based on induction over the number $k$ of different non-trivial location guards along~$\pi$.

\begin{description}
	\item[Induction base ($k=0$)] In this case $\sympath$ also exists in $\ZoneGraph(A')$, and a corresponding computation of $A'$ exists by \minreach-soundness of $\ZoneGraph(A')$. After projecting away $\gclock$, this is a computation of $A^1$ with the desired properties. %

	\item[Induction step ($k \rightarrow k+1$)] Assume that $\sympath$ has $k+1$ different non-trivial location guards, and let $(\hat{\loc},\zone_0) \xRightarrow{\tau_0,\gamma} (\loc_1,\zone_1) \xRightarrow{\tau_1,\gamma_1} \cdots \xRightarrow{\tau_{i-1},\gamma_{i-1}} (\loc_i,\zone_i)$ be the maximal prefix that has only~$k$ non-trivial location guards, \ie{} it is followed by a transition $(\loc_i,\zone_i) \xRightarrow{\tau_{i},\gamma_i} (\loc_{i+1},\zone_{i+1})$ where $\gamma_i \neq \top$ and $\gamma_i$ has not appeared as a location guard on the prefix.
Then by definition of \bzgDTN, this guarded transition can only be taken if there is a valuation $v \in \zone_i$ with $v(\gclock) \geq \nInfMinReachT{\gamma_i}$.
By induction hypothesis, there is a computation $\pi_k = \nConfigInit \xrightarrow{\delta_0,\nDisTrans{0}} \nConfig_1 \xrightarrow{\delta_1, \nDisTrans{1}} \ldots \xrightarrow{\delta_{i-1}, \nDisTrans{i-1}} \nConfig_i$ of some~$A^{n_1}$ with $v \in \zone_i$ and $v(\gclock)=\vmin(\zone_i,\gclock)$ for some $(\loc_i,v) \in \nConfig_i$.
Moreover, by (a variant of the proof of)
\cref{lemma:locationguards-ordered}, since not all guard locations that are reachable in this time may appear on the path, there exists a computation $\gcomputation_{min}$ of~$A^{k+1}$ that reaches each of the $k+1$ location guards at minimal time, and each process stays in its guard location forever.
Then the desired timed path is $\pi_k \parallel \pi_{min}$. %
\end{description}

\emph{Completeness w.r.t.~\minreach}: follows by construction of \bzgDTN and from \cref{lem:reachability}.
\qed
\end{proof}

\fi

\subsection{An Algorithm for Solving \minreach}

To solve \minreach in practice, it is usually not necessary to construct $\bzgDTN$ completely. 
For $\ReachLoc{A^\infty} = \{ \loc_1, \ldots, \loc_k\}$, assume wlog.\ that $\loc_1 \reachorder \loc_2 \reachorder \ldots \reachorder \loc_k$.
Then the timed path to $\loc_i$ with minimal global time for every $\loc_i$ can only have location guards that are in $\{ \loc_j \mid j < i \}$.
If we explore zones in order of increasing $\vmin(\zone,\gclock)$, we will find $\nInfMinReachT{\loc_1}$ without using any transitions with non-trivial location guards.
Then, we enable transitions guarded by $\loc_1$, and will find $\nInfMinReachT{\loc_2}$ using only the trivial guard and~$\loc_1$, etc.

An incremental algorithm that constructs the zone graph\footnote{As usual, for efficient zone graph construction we encode zones using \LongVersion{difference bound matrices (}DBMs, see \cite{BY03}\ifdefined\VersionForArXiV\cref{appendix:DBMs}\fi\LongVersion{)}.} %
  has to keep track of guarded transitions that cannot be taken at the current global time, but that may be possible to take at some later time.%

\cref{alg:minreachtime} takes $A$ as an input, constructs the relevant parts of $\ZoneGraph_{\UpperBound(A)}^\infty(A')$, and returns a mapping of locations $\loc$ to their $\nInfMinReachT{\loc}$.
As soon as we have discovered timed paths to all $\loc \in \Loc$, the algorithm terminates (\cref{line:terminate}).
	Otherwise, as long as we have graph nodes $(\loc,\zone)$ to explore, we pick the minimal node (w.r.t.~$\zone(\delta)$ in \cref{alg:pop}) and check that its zone $\zone$ is non-empty.
	If this is the case, we mark $\loc$ as visited and add any successor nodes.
	Furthermore, we remember any transitions that are currently not enabled, and store them with the current node in~$D$.
		
  \begin{figure}[tb]
    \centering
    \scalebox{0.9}{%
    \begin{algorithm}[H]
        \SetKwInOut{Input}{Input}
        \SetKwInOut{Output}{Output}
        \lForEach{$\loc$}{$\mathit{MinReach}(\loc)=\infty, \mathit{visited}(\loc)= \bot$}
        $W=\{(\loc_0,\zone_0)\}$, $D = \emptyset$\tcp{\textbf{W}aiting nodes and \textbf{D}isabled transitions}
  
        \While{$W \neq \emptyset \text{ and } \exists \loc. \mathit{visited}(\loc)=\bot$ \label{line:terminate}}{
          Remove $(\loc,\zone)$ from $W$ with the least $\vmin(\zone, \gclock)$ value \label{alg:pop}
  
          \If{$\zone \neq \emptyset \text{ and } \vmin(\zone, \gclock) \leq \UpperBound(A)$}{
            Set $\mathit{visited}(\loc)=\top$
            
            \tcp{compute successors along enabled transitions:}
            $S = \{ (\loc',\zone') \mid \exists \tau. \sguard{\tau}=\gamma \land (\loc,\zone) \xRightarrow{\tau} (\loc',\zone') \land (visited(\gamma) \lor \gamma=\top) \}$
            
            \tcp{new disabled transitions:}
            $D' = \{ (\tau,(\loc,\zone)) \mid  (\loc,\zone) \xRightarrow{\tau,\gamma} (\loc',\zone') \land \neg(visited(\gamma)) \}$
            
            $D = D \cup D', \quad W = W \cup S$
  
              \If{$\vmin(\zone, \delta)< \mathit{MinReach}(\loc)$\label{alg:betterZone}}{
                $\mathit{MinReach}(\loc)= \vmin(\zone, \gclock)$ \tcp{happens only once for each~$\loc$}\label{alg:storeBetterZone}
                
                \tcp{transitions enabled by $\loc$:}
                $E=\{ (\tau, (\loc',\zone')) \in D \mid \sguard{\tau}=\loc \}$
                
                \tcp{compute successors along newly enabled transitions:}
                $S = \{ (\loc'',\zone'') \mid (\tau, (\loc',\zone')) \in E$\newline
                \hspace*{0.5cm} ~ $\land ~(\loc',(\zone' \land \gclock \geq \vmin(\zone, \gclock))) \xRightarrow{\tau} (\loc'',\zone'') \}$\label{alg:enabledSuccessors}
  
                $D = D \setminus E, \quad W = W \cup S$
              }
          }
        }
      \Return{MinReach}
        
        \caption{Algorithm to solve \minreach}
        \label{alg:minreachtime}
    \end{algorithm}}
  \end{figure}
	
	Finally, if the current node $(\loc,\zone)$ is such that $\vmin(\zone, \gclock) < \mathit{MinReach}(\loc)$ (\cref{alg:betterZone}), then we have discovered the minimal reachability time of $\loc$.
	In this case we store it in $\mathit{MinReach}$(\cref{alg:storeBetterZone}), and we compute the successor along $\tau$ for every tuple $(\tau,(\loc',\zone')) \in D$ with $\sguard{\tau}=\loc$, representing a transition that has just become enabled, after intersecting $\zone'$ with $\gclock \geq \vmin(\zone, \gclock)$, as the transition is only enabled now (\cref{alg:enabledSuccessors}).

Correctness of the algorithm follows directly from \cref{lem:minreach-single}, and termination follows from finiteness of $\ZoneGraph_{\UpperBound(A)}^\infty(A')$.

\subsection{Verification of DTNs}
For a given gTA $A = (\Loc,\hat{\loc},\mathcal{C},\mathcal{T},\invariant)$, the \emph{summary automaton of $A$} is the gTA $\summary{A}=(\Loc,\hat{\loc},\mathcal{C}\cup\{\gclock\},\mathcal{\hat{T}},\invariant)$ with $\hat{\tau} = (\loc,\hat{g},r,\top,\loc') \in \mathcal{\hat{T}}$ if $\tau = (\loc,g,r,\gamma,\loc') \in \mathcal{T}$ and either $\gamma = \top \land \hat{g}=g$ or $\gamma \in \ReachLoc{A^\infty} \land \hat{g} = \left(g \land \gclock \geq \nInfMinReachT{\gamma}\right)$.

\begin{theorem}[Summary Automaton]
\label{thm:summary}
Let $A = (\Loc,\hat{\loc},\mathcal{C},\mathcal{T},\invariant) \in \famTemplates_\top$.
Then $\Lg|_{[1,i]}\left(A^\infty\right) = \Lg\Big( \big(\UG{\summary{A}} \big)^i \Big)$ for all $i \in \Nats$.
\end{theorem}

\begin{proof}
For $i=1$, the statement directly follows from \cref{lem:deltagmin} and the definition of $\summary{A}$.
For $i>1$, it follows from the fact that $\Lg|_{[j,j]}(A^\infty) = \Lg\big( \UG{\summary{A}} \big)$ for each $j \in [1,i]$. \qed
\end{proof}

\cref{thm:summary} tells us that the summary automaton $\summary{A}$ can be used to answer any verification question that is based on $\Lg|_{[1,1]}\left(A^\infty\right)$, \ie{} the local runs of a single gTA in a \DTNpersistent{} $A^\infty$.
This includes standard problems like reachability of locations, but also (global) timing properties, as well as liveness properties.
Moreover, the same holds for behaviors of multiple processes in a \DTNpersistent{}, such as mutual exclusion properties, by simply composing copies of~$\hat{A}$.
In particular, any model checking problem in a \DTNpersistent{} that can be reduced to checking a cutoff system by \cref{thm:Spalazzi} can be reduced to a problem on the summary automaton, which is exponentially smaller than the cutoff system.

	\section{Conditions for Decidability %
with Location Invariants}\label{sec:cutoffs}
In \cref{sec:intro}, we argued that our motivating example would benefit from invariants that limit how long a process can stay in any of the locations.
Neither the results presented so far nor existing cutoff results support such a model, since it would have invariants on locations that appear in location guards.
To see this, note that the correctness of \cref{thm:Spalazzi}, like the \minreach-soundness argument of the DTN-zone graph, relies on the fact that in a $DTN^-$, if we know that a location~$\loc$ can be reached (within some time or on some timed path) and we argue about the existence of a local behavior of a process, we can always assume that there are other processes in the system that reach $\loc$ \emph{and stay there forever}.
This proof construction is called \emph{flooding} of location~$\loc$, but it is in general not possible if $\loc$ has a non-trivial invariant.

In this section we generalize this flooding construction and provide sufficient conditions
under which we can obtain a reduction to a summary automaton or a cutoff system, even in the presence of invariants.
For the rest of the section, we assume that~$A$ has a single clock~$x$.

\smartpar{General Flooding Computations}
We say that $\gcomputation \in \mathcal{L}(A^\infty)$ is a \emph{flooding computation for $\loc$} if $\loc \in \finalc(\gcomputation^{\leq d})$ for every $d \geq \nInfMinReachT{\loc}$, \ie{} $\loc$ is reached in minimal time and will be occupied by at least one process at every later point in time.
Then, we obtain the following as a consequence of \cref{lem:minreach-single}:\footnote{To see this, consider the $\gcomputation_{min}$ from the soundness part of the proof, but instead of letting all processes stay in the guard locations, use our new flooding computation to determine their behaviour afterwards.}

\begin{corollary}
Let $A = (Q,\hat{\loc},\mathcal{C},\mathcal{T},\invariant)$ be a gTA, not necessarily from $\famTemplates_\top$, and let $A' = (Q,\hat{\loc},\mathcal{C}\cup\{\gclock\},\mathcal{T},\invariant)$.
If there exists a flooding computation for every $\loc \in \RGuards(A)$, then $\bzgDTN$ is correct.
\end{corollary}

Note that a flooding computation trivially exists if $\loc \in \ReachLoc{A^\infty}$ and $\invariant(\loc)=\top$.
Thus, the next question is how to determine whether flooding computations exist for locations with non-trivial invariants.

\smartpar{Identifying Flooding Computations based on Lassos}
We aim to find lasso-shaped local computations of a single process that visit a location $\loc$ infinitely often and can be composed into a flooding computation for~$\loc$.

Since any flooding computation $\gcomputation$ for $\loc$ must use one of the minimal paths to $\loc$ that we can find in $\bzgDTN$,
and any local computation in $\gcomputation$ must also be a computation of the summary automaton $\hat{A}$ of $A$,
our analysis will be based on the summary automaton $\hat{A}$ instead of~$A$.\footnote{Note that for this argument, we \emph{assume} that there are flooding computations for all $\loc \in \RGuards(A)$, and therefore $\zgDTN$ and $\hat{A}$ are correct. If we identify flooding lassos under this assumption then we have also shown that the assumption was justified.} %
Since furthermore every possible prefix of $\lcomputation$, including its final configuration $(\loc,v)$, is already computed in $\bzgDTN$, what remains to be found is the loop of the lasso, where one or more processes start from a configuration $(\loc,v)$ and always keep $\loc$ occupied.
To this end, the basic idea is that each process spends as much time as possible in~$\loc$.
To achieve this with a small number $c$ of processes, we also want to minimize times where more than one process is in~$\loc$.
We illustrate the idea on an example, and formalize it afterwards.

\begin{example}

\begin{figure}[tb]
\begin{tabular}{l  r}
  \begin{subfigure}{.49\textwidth}
  \begin{tikzpicture}[scale=1.5, font=\footnotesize]
        \begin{scope}[>=stealth]
          
          \node[location, initial] at (0, 0)  (q0) {$\hat{\loc}$};
		      \node[location] at (2, 0) (q1) {$\loc_0$};
          \node[invariant, above=of q1] {$x \leq 4$};
        \end{scope}
		\begin{scope}[->, rounded corners, >=stealth]

      \path (q0) edge[draw=brown, bend angle=10, bend left =25, thick] node[above]{$x=2$} (q1);
      \path (q1) edge[draw=red, bend angle=10, bend left =25, thick]
      node[above]{$x \leftarrow 0$} node[below, locguard]{$\loc_0$} (q0);
		\end{scope}
		\end{tikzpicture}
  \caption{A loop that can flood location $\loc_0$}
  \label{fig:gtp-timed-loop}
  \end{subfigure} \vspace*{0.1cm}
&
\begin{subfigure}{.49\textwidth}
    \centering
     \begin{tikzpicture}[scale=0.5, font=\footnotesize]
        \begin{scope}[>=stealth]
          
          \node [location, initial] (q0) at (-3, -3) { $\hat{q}$};
          \node [location] (q1) at (0,-3) {  $\loc_0$ };
          \node [location] (q2) at (-3,-4.5) {  $\loc_{i+1}$ };
          \node [location] (q3) at (3,-4.5) {$\loc_{j+1}$ };
          \node  (fake) at (0,0) {};
          \node[invariant, above= of q1] {$x \leq d$};
        \end{scope}
		\begin{scope}[->, rounded corners, >=stealth]
      \path[->,thick,dashed](q0) edge [] (q1);

      \draw[thick, dashed, draw=red](q1)--(q2);
      \node[] at (-1.5, -4.2) {\tiny $x \leftarrow 0$};
		
		\draw[-,thick, dashed,draw=ForestGreen](q2)--(-1.25,-4.5);
		\draw[->, thick,draw=ForestGreen](-1.2,-4.5)--(0.2,-4.5);
    \node[] at (-0.6, -4.8) {\tiny $x \leftarrow 0$};
		\draw[-,thick, dashed,draw=ForestGreen](0.3,-4.5)--(1,-4.5);
		\path[->, thick,draw=ForestGreen](1,-4.5) -- (q3);
    \node[] at (1.4, -4.8) {\tiny $x \leftarrow 0$};

		\path[->, thick, dashed, draw=brown](q3)--(q1);

		\end{scope}
  \end{tikzpicture}
  \caption{A general lasso for flooding~$\loc_0$, showing $\psi_1$ in red, $\psi_2$ in green, and $\psi_3$ in brown}
  \label{fig:general-timed-loop-color-codings}
\end{subfigure}
\end{tabular}
\caption{Loops and color codings for flooding computations}
\label{fig:gtp-timed-loop-with-general-version}

\end{figure}
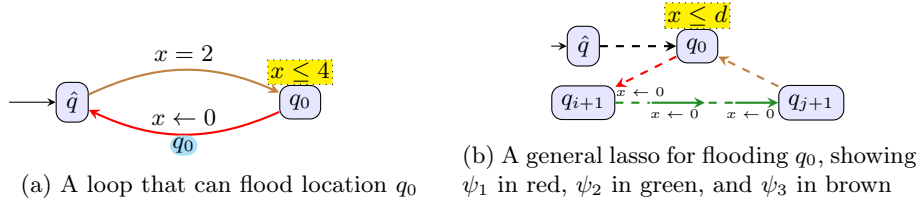

\label{ex:flooding_loop}
\cref{fig:gtp-timed-loop} shows a subgraph of the gTA in \cref{figure:example-gTA}.
To find a flooding computation for $\loc_0$,
we let two processes $p1$ and $p2$ start in location~$\hat{\loc}$ and have both of them move to~$\loc_0$ at global time~$2$.
Then, $p1$ immediately moves to~$\hat{\loc}$, which is possible as the location guard is enabled by $p2$.
After $2$ time units, $p2$ has to leave $\loc_0$ due to the invariant, so we let $p1$ return to $\loc_0$, which allows $p2$ to take the transition to $\hat{q}$, where stays for $2$ time units.
In this way, both processes can keep taking turns traversing the loop, keeping $\loc_0$ occupied forever.

\end{example}

\smartpar{Syntactic Paths and Timed Loops}
A (finite) \emph{path} in $\hat{A}$ %
is a sequence of transitions $\psi = \loc_0 \xrightarrow{\tau_0} \loc_1 \xrightarrow{\tau_1} \ldots \xrightarrow{\tau_{l-1}} \loc_{l}$, with $\tau_i=(\loc_i,g_i,r_i,\gamma_i,\loc_i') \in \mathcal{T}$ and $\loc_i'=\loc_{i+1}$ for $i \in [0,l[$.
We also denote $\loc_0$ as $\firstq(\psi)$, and call $\psi$ \emph{initial} if $\firstq(\psi) = \hat{\loc}$.
A path $\psi$ is a \emph{loop} if $\loc_0=\loc_l$.
We call a path \emph{resetting} if it has at least one clock reset.

We call a path $\psi = \loc_0 \xrightarrow{\tau_0} \ldots \xrightarrow{\tau_{l-1}} \loc_{l}$ \emph{executable from $(\loc_0,v_0)$} if there is a timed path $\lcomputation = (\loc_0,v_0) \xrightarrow{\delta_0,\tau_0}  \ldots \xrightarrow{\delta_{l-1},\tau_{l-1}} (\loc_{l},v_l)$ of $\hat{A}$.
We say that $\lcomputation$ is \emph{based on}~$\psi$.
For a path $\psi = \loc_0 \xrightarrow{\tau_0} \ldots \xrightarrow{\tau_{l-1}} \loc_{l}$ in $\hat{A}$ that is executable from $(\loc_0,v_0)$, we denote by  $\asap(\psi,v_0)$ the unique timed path $(\loc_0,v_0) \xrightarrow{\delta_0,\tau_0}  \ldots \xrightarrow{\delta_{l-1},\tau_{l-1}} (\loc_{l},v_l)$ such that in every step, $\delta_i$ is minimal.
A timed path $\lcomputation = \ourloop$ is called a \emph{timed loop} if $\loc_0=\loc_l$ and $u_0(x)=u_l(x)$.

\smartpar{Sufficient Conditions for Existence of a Flooding Computation}
Let $\psi$ be a resetting loop, 
where $\tau_i$ is the first resetting transition and $\tau_j$ is the last resetting transition on $\psi$.
As depicted in \cref{fig:general-timed-loop-color-codings}, we split $\psi$ into %
\begin{oneenumerate}%
    \item $\psi_1 = \psione$;
    \item $\psi_2=\psitwo$; and
    \item $\psi_3=\psithree$.
\end{oneenumerate}
Note that $\psi_2$ is empty if $i=j$, and $\psi_3$ is empty if $j=l-1$. 
If $\psi$ is executable from $(\loc_0,v_0)$, then $\asap(\psi,v_0) = (\loc_0,v_0) \xrightarrow{\delta_0,\tau_0}  \ldots \xrightarrow{\delta_{l-1},\tau_{l-1}} (\loc_{l},v_l)$ exists and we can compute the time needed to execute its parts as
\begin{itemize}
\item $d_1= \totaltime\big((\loc_0,v_0) \xrightarrow{\delta_0,\tau_0}  \ldots \xrightarrow{\delta_{i},\tau_{i}} (\loc_{i+1},v_{i+1})\big)$,
\item $d_2= \totaltime\big((\loc_{i+1},v_{i+1}) \xrightarrow{\delta_{i+1},\tau_{i+1}}  \ldots \xrightarrow{\delta_{j},\tau_{j}} (\loc_{j+1},v_{j+1})\big)$,
\item $d_3= \totaltime\big((\loc_{j+1},v_{j+1}) \xrightarrow{\delta_{l-1},\tau_{l-1}}  \ldots \xrightarrow{\delta_{l-1},\tau_{l-1}} (\loc_{0},v_{l})\big)$.
\end{itemize}

Note that the first process $p_1$ that traverses the loop along $\asap(\psi,v_0)$ will return to $\loc_0$ after time $d_1+d_2+d_3$.
Thus, if a process $p_2$ starts in the same configuration $(\loc_0,v_0)$, and stays in $\loc_0$ while $p_1$ traverses the loop, then $p_1$ will return to $\loc_0$ before $p_2$ has to leave due to $\invariant(\loc_0)$ if $U^\invariant_x(\loc_0) \geq d_1+d_2+d_3+v(x)$, where $U^\invariant_x(\loc_0)$ denotes the upper bound imposed on $x$ by $\invariant(\loc_0)$.
Moreover, $p_2$ can still execute $\psi$ from a configuration $(\loc_0,v')$ with $v'=v+\delta$ if and only if $v(x)+\delta \leq T$, where $T = \min\{ U^\invariant_x(\loc_k) \mid 0 \leq k \leq i\}$.
Generally, traversing $\psi$ is possible from any configuration $(\loc_0,v')$ with $v(x) \leq v'(x) \leq T$ and $v(\gclock) \leq v'(\gclock)$. In particular, we have $\delta(\asap(\psi,v')) \leq d_1+d_2+d_3$ and in the reached configuration $(\loc_0,v'')$ we have $v''(x)\leq d_3$.
Thus, if $p_2$ has to leave before $p_1$ returns, we can add more processes that start traversing the loop in regular intervals, such that $p_{i+1}$ always arrives in $\loc_0$ before $p_i$ has to leave.

\begin{example}
Consider again the flooding computation constructed in \cref{ex:flooding_loop}, depicted in \cref{fig:gtp-flooding-global-path-with-Invariants}.
In this case, $\psi_1 = \loc_0 \rightarrow \hat{\loc}$ (depicted in red), $\psi_2$ is empty, and $\psi_3 = \hat{\loc} \rightarrow \loc_0$ (in brown).
Note that we have a repetition of the global configuration after a bounded time, and therefore can keep $\loc_0$ occupied forever by repeating the timed loop.
\usetikzlibrary{decorations.pathreplacing}

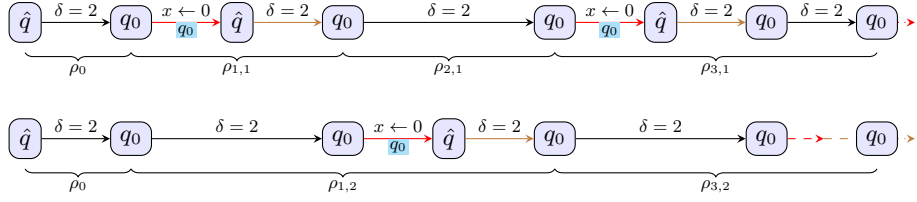
\begin{figure}[tb]
  
\begin{tikzpicture}[ font=\footnotesize]

    \matrix[column sep=0.9cm, row sep=1cm] (m) {
        \node[location](q0_0){$\hat{\loc}$}; 
      & \node[location] (q1_0){$\loc_0$}; 
      & \node[location] (q2_0){$\hat{\loc}$};
      & \node[location] (q3_0){$\loc_0$}; 
      & 
      & \node[location] (q5_0){$\loc_0$};
      & \node[location] (q6_0){$\hat{\loc}$};
      & \node[location] (q7_0){$\loc_0$};
      & \node[location] (q8_0){$\loc_0$};
      \\
        \node[location](q0_1){$\hat{\loc}$}; 
      & \node[location] (q1_1){$\loc_0$}; 
      &  
      & \node[location] (q3_1){$\loc_0$}; 
      & \node[location] (q4_1){$\hat{\loc}$}; 
      & \node[location] (q5_1){$\loc_0$};
      &
      & \node[location] (q7_1){$\loc_0$};
      & \node[location] (q8_1){$\loc_0$};
      \\
    };

    \begin{scope}[->,-stealth, every node/.style={scale=0.75}]
      \path[] (q0_0) edge[] node[above]{$\delta = 2$}  (q1_0);
      \path[] (q0_1) edge[] node[above]{$\delta = 2$}  (q1_1);

      \path[draw=red, red] (q1_0) edge[] node[above]{\color{black}$x
      \leftarrow 0$} node[below, locguard]{\color{black}$q_0$}  (q2_0);

      \path[draw=brown, brown] (q2_0) edge[] node[above]{\color{black}$\delta = 2$} (q3_0);

      \path[] (q1_1) edge[] node[above]{ $\delta = 2$} (q3_1);

      \path[draw=red, red] (q3_1) edge[] node[above]{\color{black}$x \leftarrow 0$} node[below, locguard]{\color{black}$q_0$}  (q4_1);

      \path[draw=brown, brown] (q4_1) edge[] node[above]{\color{black}$\delta = 2$} (q5_1);

      \path[] (q3_0) edge[] node[above]{ $\delta = 2$} (q5_0);

      \path[draw=red, red] (q5_0) edge[] node[above]{\color{black}$x \leftarrow 0$} node[below, locguard]{\color{black}$q_0$}  (q6_0);

      \path[draw=brown,brown] (q6_0) edge[] node[above]{\color{black}$\delta = 2$} (q7_0);

      \path[] (q5_1) edge[] node[above]{$\delta = 2$} (q7_1);

      \path[] (q7_0) edge[] node[above]{$\delta = 2$} (q8_0);
      \path[dashed, draw=red, red] (q7_1) edge[] node[above]{} ($(q8_1) -  (0.7,0)$);
      \path[-,dashed, draw=brown, brown] ($(q8_1) -  (0.7,0)$) edge[] node[above]{} (q8_1);

      \path[dotted, draw=red, red] (q8_0) edge[] ($(q8_0) + (0.5,0)$);
      \path[draw=brown, dotted, brown] (q8_1) edge[] ($(q8_1) + (0.5,0)$);
    \end{scope}

    \begin{scope}[every node/.style={scale=0.75}]
      \draw[decorate, decoration={brace,amplitude=3pt,mirror}] ($ (q0_0) - (0,0.4)$) -- ($ (q1_0) - (0,0.4)$) node[midway, black, yshift=-0.3cm] {$\rho_0$};
      \draw[decorate, decoration={brace,amplitude=3pt,mirror}] ($ (q1_0) - (0,0.4)$) -- ($ (q3_0) - (0,0.4)$) node[midway, black, yshift=-0.3cm] {$\rho_{1,1}$};
      \draw[decorate, decoration={brace,amplitude=3pt,mirror}] ($ (q3_0) - (0,0.4)$) -- ($ (q5_0) - (0,0.4)$) node[midway, black, yshift=-0.3cm] {$\rho_{2,1}$};
      \draw[decorate, decoration={brace,amplitude=3pt,mirror}] ($ (q5_0) - (0,0.4)$) -- ($ (q8_0) - (0,0.4)$) node[midway, black, yshift=-0.3cm] {$\rho_{3,1}$};

      \draw[decorate, decoration={brace,amplitude=3pt,mirror}] ($ (q0_1) - (0,0.4)$) -- ($ (q1_1) - (0,0.4)$) node[midway, black, yshift=-0.3cm] {$\rho_0$};
      \draw[decorate, decoration={brace,amplitude=3pt,mirror}] ($ (q1_1) - (0,0.4)$) -- ($ (q5_1) - (0,0.4)$) node[midway, black, yshift=-0.3cm] {$\rho_{1,2}$};
      \draw[decorate, decoration={brace,amplitude=3pt,mirror}] ($ (q5_1) - (0,0.4)$) -- ($ (q8_1) - (0,0.4)$) node[midway, black, yshift=-0.3cm] {$\rho_{3,2}$};
    \end{scope}
    
  \end{tikzpicture}

  \caption{Flooding computation for location $\loc_0$ from  \cref{fig:gtp-timed-loop}}
  \label{fig:gtp-flooding-global-path-with-Invariants}
\end{figure}

\end{example}

Based on these observations, we can show that the conditions mentioned above are sufficient to guarantee that a flooding computation for $\loc_0$ exists, provided that we also have a flooding computation for all other guard locations.

\begin{restatable}{lemma}{floodingloop}

\label{lem:flooding_loop}
Let $\lcomputation_0 = (\hat{\loc},\mathbf{0})  \xrightarrow{} \ldots \xrightarrow{} (\loc_0,v_0)$ be a computation of $\hat{A}$ that reaches $\loc_0$ after minimal time, and $\psi = \loc_0 \xrightarrow{\tau_0} \ldots \xrightarrow{\tau_{l-1}} \loc_0$ a resetting loop in~$\hat{A}$ executable from $(\loc_0,v_0)$.
Let $d_1,d_2,d_3, T$ as defined above.
If $T \geq d_1 + d_2 + d_3 + v(x)$, $T > d_3$, and there exists a flooding computation for all $\loc \in \RGuards(A){\setminus}\{\loc_0\}$, then there exists a flooding computation for $\loc_0$.

\end{restatable}

\iffull
\begin{proof}
To show that there exists a flooding computation for $\loc_0$ in $\mathcal{L}(A^\infty)$, we construct local timed paths of $\hat{A}$ for a number $c$ of processes, and then show that there exists a computation $\pi$ of $A^n$ with $n \geq c$ such that $\pi|_{[1,c]}$ is the composition of these local timed paths, and therefore $\pi$ is a flooding computation of $\loc_0$.
Note that in both cases below, all location guards except $\loc_0$ are implicitly taken care of, since $\rho_0$ and our computation of the $t_i$ are based on $\hat{A}$ (and therefore take the minimal time to reach guard locations into account), and by assumption all other locations have a flooding path.

\textbf{Case 1}: {$d_2=d_3=0$}, \ie{} if $x=T$ then the loop can be traversed from $\loc_0$ in $0$ time.
In this case the flooding computation lets two processes move to $\loc_0$ in minimal time, lets both of them stay in $\loc_0$ until $x=T$, then lets process~1 traverse the loop in $0$ time, and finally does the same for process~2.
This can be repeated until infinity, always keeping $\loc_0$ occupied.
Note that a flooding computation with a single process may not be possible, as $\loc_0$ could appear as a location guard.

\textbf{Case 2}: The loop cannot be traversed in $0$ time.
In this case, let {$c = \lceil \frac{T+d_2}{T-d_3} \rceil$}, and
for each $p \in [1,c]$, define $\rho(p) = \rho_{1, p} \cdot \rho_{2, p} \cdot \rho_{3, p}$ as follows:
\begin{itemize}
\setlength\itemsep{-0.2em}
    \item $\rho_{1, p} = (\loc_0,v_0) \xrightarrow{\delta_p} (\loc_0, v_0 + \delta_p) \cdot \asap(\psi,v_0 + \delta_p)$ with {$\delta_p = (p-1)(T-d_3)$}
		\item $\rho_{2, p} = (\loc_0,v_1) \xrightarrow{\delta'_p} (\loc_0, v_1 + \delta'_p)$ with $v_1$ the resulting clock valuation of $\asap(\psi,v_0 + \delta_p)$ and $\delta'_p$ such that $(v_0+\delta'_p)(x)=T$
    \item $\rho_{3, p} = \asap(\psi,v_2) \cdot (\loc_0,v_2') \xrightarrow{\delta''_p} (\loc_0, v_2' + \delta''_p)$ with $v_2 = v_1 + \delta'_p$, $v_2'$ the resulting clock valuation of $\asap(\psi,v_2)$, and $\delta''_p$ such that $(v_2'+\delta''_p)(x)=T$
\end{itemize}

By construction, $\rho(p) \in \mathcal{L}(\hat{A})$ for every $p$. 
Moreover, note that, up to the valuation of $\gclock$, $\rho_{3,p}$ can be appended an arbitrary number of times to $\rho(p)$ (since for any $v$ with $v(x)=T$ and {$v(\gclock) \geq v_0(\gclock)+d_1+d_2+d_3$}, the time $\delta(\asap(\loc_0,v))$ and the resulting clock valuation for $x$ will always be the same).
Thus, let $\rho^\omega(p)$ be the projection of $\rho_{1,p} \cdot \rho_{2,p} \cdot (\rho_{3, p})^\omega$ to $A$, i.e., without valuations of $\gclock$.

Now, note that at any time {$d \geq 0$}, $\loc_0$ is occupied in at least one of the $\rho^\omega(p)$: for {$d \leq d_1 + d_2 + d_3$}, process~$c$ has not left $\loc_0$ yet, and
at any time {$d > d_1 + d_2 + d_3$}, process {$\lceil \frac{d - (d_1 + d_2 + d_3)}{T - d_3} \rceil \mod c$} occupies $\loc_0$.

\textbf{Both cases}: 
The desired $\pi$ is the composition of $\rho_0 \cdot \rho^\omega(p)$ with the flooding computations for all $\loc \in \ReachLoc{A^\infty}{\setminus}\{\loc_0\}$, which we assumed to exist.
\qed
\end{proof}

\fi

{Note that this property depends on our assumption that $A$ has a single clock. It remains open whether a more complex construction works for multiple clocks.}

\begin{restatable}{lemma}{mutualflooding}

\label{lem:mutualflooding}
If for every $\loc \in \RGuards(A)$ we either have $\invariant(\loc)=\top$ or there exists a flooding computation, then there exists a computation of $A^\infty$ that floods all of these locations at the same time.
\end{restatable}

\ifdefined\VersionForArXiV{
    \iffull
    \begin{proof}
To arrive at a contradiction, assume that for every $\loc \in \RGuards(A)$ we have a flooding computation $\pi_q$ according to \cref{lem:flooding_loop}, %
but in their composition there is a $\loc_0$ such that $\loc_0$ is not occupied at some time {$d \geq \nInfMinReachT{\loc_0}$}.
W.l.o.g., let~$d$ be the smallest global time where this happens, for any of the guard locations.
By construction, the only reason for this can be that on the flooding loop of $\loc_0$ we have some $\loc' \neq \loc_0$ as a location guard.
But since $\pi_{\loc_0}$ has been constructed based on $\hat{A}$, this transition must happen after $\nInfMinReachT{\loc'}$.
Then, this would mean that the flooding construction would have failed for $\loc'$ before it has failed for~$\loc_0$, contradicting our assumption that {$d$} is minimal.
 \qed
\end{proof}

    \fi
}\fi

\smartpar{New Cutoff Results}
In addition to witnessing the correctness of the summary automaton for $A$, the proofs of \cref{lem:flooding_loop,lem:mutualflooding} also allow us to compute a cutoff for the given gTA.
For a location $\loc$, let $w(\loc) = 1$ if $\invariant(\loc)=\top$ and $w(\loc) = max\{2, \lceil \frac{T+d_2}{T-d_3} \rceil\}$, where $T, d_2, d_3$ are as in \cref{lem:flooding_loop}, if there exists a flooding computation based on the lemma.
Intuitively, $w(\loc)$ is the \emph{width} of the flooding computation, \ie{} how many processes need to be dedicated to~$\loc$\LongVersion{ in this computation}.
\LongVersion{Then we get the following cutoff result.
}

\begin{corollary}
For any gTA $A$ that satisfies the conditions of \cref{lem:mutualflooding} and any $m \in \Nats$, $m+ \sum_{\loc \in \RGuards(A)} w(\loc)$ is a cutoff.
\end{corollary}

\smartpar{Sufficient and Necessary Conditions for Decidability}
Note that above, we give sufficient but not necessary conditions to establish that a guard location $\loc$ can always remain occupied after $\nInfMinReachT{\loc}$.
Further note that there are {\DTN}s where a guard location $\loc$ \emph{cannot} always remain occupied after it is reached, and in such cases the language $\mathcal{L}|_{[1,i]}\left(A^\infty\right)$ can be determined iff one can determine all (global-time) intervals in which $\loc$ can be occupied in the \DTN, for all such $\loc$.
While it is known that in this case cutoffs for parameterized verification do not exist~\cite{SpalazziS20}, an approach based on a summary automaton would work whenever these intervals can be determined.
Whether parameterized verification is decidable in the presence of location invariants in general remains an open question.

	\section{Evaluation}\label{section:evaluation}
We compare the performance of parameterized verification of {\DTN}s based on our new techniques to the existing cutoff-based techniques (according to \cref{thm:Spalazzi}, using \Uppaal{}).
To this end, we implemented \cref{alg:minreachtime} and the detection of flooding lassos from \cref{sec:cutoffs}, and constructed three parametric examples:
\begin{ienumerate}
\item parametric versions of the TA in \cref{fig:gtpProtocol} with locations $h_0,\cdots,h_{k-1}, \ell_0,\cdots,\ell_{k-1}, h_{sy}, \ell_{sy}$ for some $k$, without location invariants on the $h_i$\ifdefined\VersionForArXiV (see \cref{sec:extended-example} for $k=3$)\fi, denoted $GCS^\top(k)$ in the following,
\item versions of $GCS^\top(k)$ with invariants on the $h_i$, denoted $GCS(k)$,
\item hand-crafted examples $Star(k)$, where some $q_{\mathsf{final}}$ can only be reached after all $q \in \mathit{Guards}(A)$ with $\card{\mathit{Guards}(A)}=k$ have been reached \ifdefined\VersionForArXiV(see \cref{sec:star-example} \fi{}for $k=4$ and $5$).
\end{ienumerate}

All experiments have been conducted on a machine with Intel$\textregistered$ Core\texttrademark{} i7-8565U
CPU @ 1.80GHz and 16\,GiB RAM, and our implementation as well as the benchmarks
can be found at
\url{https://doi.org/10.5281/zenodo.8337446}.\footnote{For the latest version of the tool refer to the \href{https://github.com/cispa/Verification-Disjunctive-Time-Networks}{GitHub repository}.}

\begin{table}[t]
\caption{Performance comparisons, all times in seconds, timeout (TO) is 3000\,s.}

\begin{subtable}{.35\textwidth}
\caption[position=above]{\footnotesize Comparison of \cref{alg:minreachtime} and \Uppaal{} (on the cutoff system $A^c$) for detection of minimal reachability times on $Star(k)$.}
\label{tab:reachtime}
\scalebox{1}{
\scriptsize
\begin{tabular}{cc|rr}
\toprule
\multicolumn{2}{c|}{\textbf{Benchmark}} &  \multicolumn{2}{c}{\textbf{Time}}  \\
Name & Cutoff $c$ & Alg.~1 & ~~\Uppaal{}\\
\midrule
$Star(4)$ & 5 & ~$4.8$ & ~$36.0$  \\
$Star(5)$ & 6 & $9.1$ & TO\\
$Star(6)$ & 7 & $15.7$ & TO\\

\bottomrule
\end{tabular}
}

\end{subtable}
\hspace{.02\textwidth}
\begin{subtable}{.58\textwidth}
\caption[position=above]{\footnotesize Comparison of verification times for different versions of \cref{fig:gtpProtocol} and properties $\phi_i$, based on the summary automaton ($\hat{A}$) or the cutoff system ($A^c$).}
\label{tab:performance}
\scalebox{1}{
\scriptsize
\begin{tabular}{ccc|cc|rr}
\toprule
\multicolumn{3}{c|}{\textbf{Benchmark}} &  & & \multicolumn{2}{c}{\textbf{Time}}  \\
Name & $\left| Q \right|$ & Cutoff $c$ & Property & Result & $\hat{A}$~~~&  $A^c$~\\
\midrule
\multirow{2}{*}{$GCS^\top(3)$} & \multirow{2}{*}{8} & \multirow{2}{*}{4}&$\phi_1$& False & ~~$< 0.1$ & ~~$< 0.1$  \\

  & &  &$\phi_3$ & True  & $< 0.1$ & 26.0  \\
    
 \midrule
\multirow{2}{*}{$GCS(3)$} & \multirow{2}{*}{8}  & \multirow{2}{*}{7} &$\phi_1$ & False & $< 0.1$ & $< 0.1$ \\
      
  & & &$\phi_3$ & True &  $< 0.1$ & TO  \\

\bottomrule
\end{tabular}
}
\end{subtable}
\end{table}

On $Star(k)$, \cref{alg:minreachtime} significantly outperforms the cutoff-based approach for solving \minreach with \Uppaal{}, as can be seen in \cref{tab:reachtime}. 
On $GCS^\top(3)$, solving \minreach and constructing the summary automaton takes 0.23\,s, and 1.13\,s for $GCS(3)$.
Solving \minreach using cutoffs is even faster, which is not surprising since in this example location guards do not influence the shortest paths.
However, we can also use the summary automaton to check more complex temporal properties of the \DTN, and two representative examples are shown in \cref{tab:performance}: $\phi_1$ states that a process that is in a location $h_i$ will eventually be in a location $\ell_j$, and $\phi_3$ states that a process 
in a location $h_i$ will eventually be in a location $h_j$.
For $\phi_1$, both approaches are very fast, while for $\phi_3$ our new approach significantly outperforms the cutoff-based approach.
Note that even if we add the time for construction of the summary automaton to the verification time, we can solve most queries significantly faster than the cutoff-based approach\LongVersion{ (and in fact we only need to construct the summary automaton once for every benchmark, and then can solve all verification queries for that benchmark quickly)}.
Additional experimental results can be found in \ifdefined\VersionForArXiV\cref{sec:experiments-full}\else\cite{AEJK23report}\fi.

	\section{Conclusion}\label{section:conclusion}
In this work, we proposed a novel technique for parameterized verification of disjunctive timed networks ({\DTN}s), \ie{} an unbounded number of timed automata, synchronizing on disjunctive location guards. %
Our technique to solve \minreach in a network of arbitrary size relies on an extension of the zone graph of a \emph{single} TA of the network---leading to an exponential reduction of the model to be analyzed, when compared to classical cutoff techniques.

If guard locations can always remain occupied after first reaching them, solving \minreach allows us to construct a \emph{summary automaton} that can be used for parameterized verification of more complex properties, which is again exponentially more efficient than existing cutoff techniques.
This is the case for the full class of {\DTN}s without invariants on guard locations, and we give a sufficient condition for correctness of the approach on the larger class with invariants, but with a single clock per automaton.
Moreover, our \emph{ad-hoc} prototype implementation already outperforms cutoff-based verification in \Uppaal{} on tasks that significantly rely on location guards.

Decidability of the parameterized verification problem for {\DTN}s with
invariants on guard locations in general remains an open question, but the
techniques we introduced are a promising step towards solving it in cases where
it is known that cutoffs do not exist~\cite{SpalazziS20}.

	\section*{Data Availability}

		The program (including the source code) and benchmark files as evaluated in \cref*{section:evaluation} are available at \url{https://doi.org/10.5281/zenodo.8337446}. Additionally, we plan to track all future development of our tool in the GitHub repository at \url{https://github.com/cispa/Verification-Disjunctive-Time-Networks}.

	\bibliographystyle{splncs04}
	\bibliography{main}

	\ifdefined\VersionForArXiV
		\appendix
		\newpage
		\iffull
\else
\input{proofs}
\fi

\section{DBMs: Matrix Representation of Zones}\label{appendix:DBMs}

For a set of clocks $\mathcal{C}$, we denote by $\mathcal{C}_0 = \mathcal{C} \cup \{c_0\}$ the set $\mathcal{C}$ extended with a special variable~$c_0$ with the constant value~$0$.
For convenience we sometimes write~$0$ to represent the variable~$c_0$.

A \emph{difference bound matrix (DBM)} for a set of clocks $\mathcal{C}$ is a $|\mathcal{C}_0| \times |\mathcal{C}_0|$-matrix $(Z_{xy})_{x,y \in \mathcal{C}_0}$, in which each entry $Z_{xy}= (\lessdot_{xy},c_{xy})$ represents the constraint $x-y \lessdot_{xy} c_{xy}$ where $c_{xy} \in \Ints$ and $\lessdot_{xy} \in \{<, \leq\}$ or $(\lessdot_{xy},c_{xy}) = (<, \infty)$.
Then, a DBM $(Z_{xy})_{x,y \in \mathcal{C}_0}$ represents the zone $\bigwedge_{x,y \in \mathcal{C}_0} (Z_{xy})$.

A DBM is \emph{canonical} if none of its constraints can be strengthened without reducing the set of solutions.
Given a $n \times n$-DBM with a non-empty set of solutions, the canonical DBM with the same solutions can be calculated in time $\mathcal{O}(n^3)$.
For more details on DBMs, %
see~\cite{BY03}.

\section{Extended Example $GCS(3)$}
\label{sec:extended-example}

\begin{figure}[h]
  \scalebox{1}{
			\begin{tikzpicture}[scale=0.7, font=\footnotesize]
        \begin{scope}[]
          
          \node [location] (q1) at (1, -2) { $h_{sy}$ };
          \node [location] (q8) at (1,-4) { $\ell_{sy}$ };
          \node [location] (q2) at (4,0) {  $h_0$ };
          \node [location] (q3) at (9,0) {  $h_1$ };
          \node [location] (q4) at (14,0) {  $h_2$ };
          \node [location] (q5) at (4,-6) { $\ell_0$ };
          \node [location] (q6) at (9,-6) { $\ell_1$ };
          \node [location] (q7) at (14,-6) {  $\ell_2$ };
          
          \node [invariant] at (5,-0.5) {$x \leq 2$};
          \node [invariant] at (10,-0.5) {$x \leq 2$};
          \node [invariant] at (15,-0.5) {$x \leq 2$};
          \node [invariant] at (5,-6.5) {$x \leq 4$};
          \node [invariant] at (10,-6.5) {$x \leq 4$};
          \node [invariant] at (15,-6.5) {$x \leq 4$};
          
          \node [rotate=35]  (t1) at (2.55,-2.65)   {$x \leftarrow 0$};
          \node [locguard,rotate=35]  (t1) at (2.8,-3.1)   {$h_0$};
          \node [rotate=14]  (t1) at (4.5,-2.9)   {$x \leftarrow 0$};
          \node [locguard,rotate=14]  (t2) at (4.5,-3.5)   {$h_1$};   
          \node []  (t1) at (5.7,-3.7)   {$x \leftarrow 0$};      
          \node [locguard,] (t3) at (5.7, -4.3){$h_2$};
          
          \node []  (t1) at (2.2,1.2)   {$x \leftarrow 0$};
          \node [locguard] (t4) at (2.2,0.7)  {$h_0$};
          \node []  (t1) at (4,2)   {$x \leftarrow 0$};
          \node [locguard] (t5) at (4,1.5)      {$h_1$};
          \node []  (t1) at (6.5,2.75)   {$x \leftarrow 0$};
          \node [locguard] (t6) at (6.5,2.25){$h_2$};
         
         \node [rotate = -50] (t4) at (12.8, -4.2){$x \neq 2, x \leftarrow 0$};
         \node [locguard, rotate = -50] (t4) at (12.4, -4.6){$h_0, h_1, h_2$};
          
         \node [rotate = 35]  (t4) at (8.85,-2.75)    {$x \neq 2, x \leftarrow 0$};
         \node [locguard,rotate = 32]  (t4) at (9.5,-3.1)    {$h_0, h_1, h_2$};
         
         \node [rotate = -50]  (t4) at (5.95,-2)    {$x \neq 2, x \leftarrow 0$};
         \node [locguard,rotate = -50]  (t4) at (5.3,-2.1)    {$h_0, h_1, h_2$};
         
         \node [] at (9,1.8) {$x = 2, x \leftarrow 0$};
         \node [] at (6.5,.2) {$x = 2, x \leftarrow 0$};
         \node [] at (11.5,.2) {$x = 2, x \leftarrow 0$};
         \node [] at (9,-7.5) {$x \geq 1,x \leftarrow 0$};
         \node [] at (6.5,-5.8) {$x \geq 1,x \leftarrow 0$};
         \node [] at (11.5,-5.8) {$x \geq 1,x \leftarrow 0$};

          \node  (fake) at (2.5, 0) {};
                  
        \end{scope}

    \begin{scope}[->, rounded corners, >=stealth]
     \draw (q2)--(q3);
     \draw (q3)--(q4);
     \draw (q4)[rounded corners=10pt]--(9,1.75)--(q2);
     \draw (q5)--(q6);
     \draw (q6)--(q7);
     \draw (q4)--(q5);
    
     \draw(q7)[rounded corners=10pt]--(14,-7.75)--(4,-7.75)--(q5);
     
     \path[] (q5) edge[bend right] node[] {} (q8);
     \path[] (q8) edge[bend right] node[] {} (q5);
     
     \draw (q1) [rounded corners=10pt] -- (1.5,1) -- (3,1)  -- (q2) ;
     \draw (q1) [rounded corners=10pt] -- (1.25,1.75) -- (5,1.75) [sharp corners] -- (q3) [rounded corners=5pt] ;
     \draw (q1) [rounded corners=10pt] -- (1,2.5) -- (14,2.5) [sharp corners] -- (q4)[rounded corners=5pt] ;
     
     \draw(q2)[rounded corners=10pt] --(q1);
     
     \draw[](q2)--(q6);

     \draw[](q3)--(q7);

       \draw (q8) [rounded corners=10pt]  --(4,-2)  --    (q2)   [rounded corners=5pt] ;            
       \draw (q8) [rounded corners=10pt]   -- (7.5,-2.5)-- (q3)   [rounded corners=5pt] ;
       \draw (q8) [rounded corners=10pt]   -- (9.5,-4)-- (q4)   [rounded corners=5pt] ;
       \draw (fake) edge (q2); 
    \end{scope}

\end{tikzpicture}}
  \end{figure}

\newpage
\section{Crafted Example $Star(k)$}
\label{sec:star-example}

\begin{figure}[h]
\includegraphics[width=.7\textwidth]{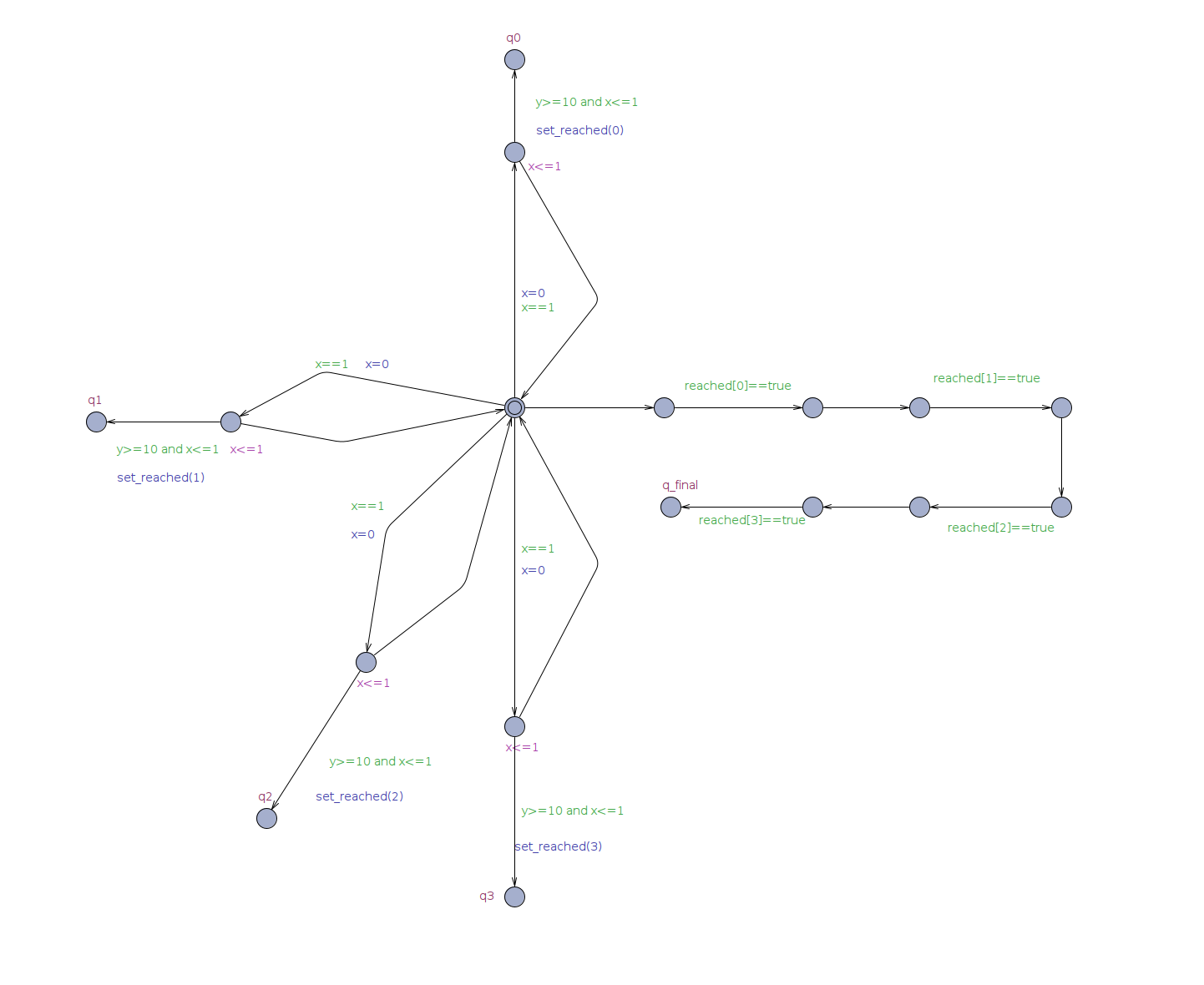}
\vspace*{-0.75cm}
\caption{Crafted Example $Star(4)$}
\end{figure}

\begin{figure}[h]
\includegraphics[width=.66\textwidth]{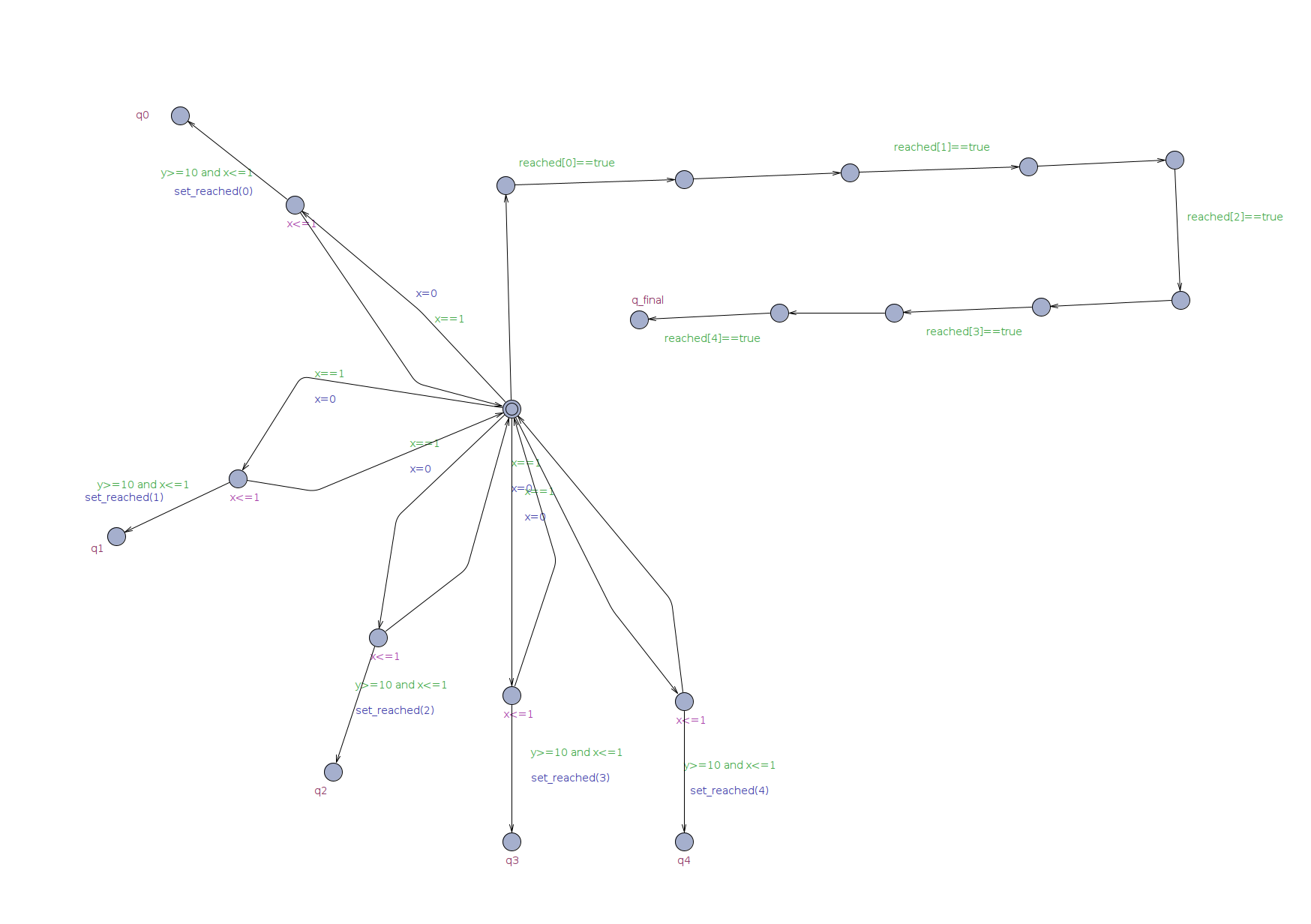}
\vspace*{-0.5cm}
\caption{Crafted Example $Star(5)$}
\end{figure}

\newpage
\section{Full Experiments for $GCS$ Examples}
\label{sec:experiments-full}
The construction of the summary automaton took 0.23\,s for $GCS^\top(3)$, 1.13\,s for $GCS(3)$, 0.86\,s for $GCS^\top(4)$, and 5.50\,s for $GCS(4)$.

\cref{tab:performance-full} gives the results of checking five properties $\phi_i$ on $GCS^\top(k)$ and $GCS(k)$ with $k \in \{3,4\}$, where:
\begin{ienumerate}
\item $\phi_1$ states that a process that is in a location $h_i$ will eventually be in a location $\ell_j$, 
\item $\phi_2$ states that a process that is in a location $\ell_i$ will eventually be in a location $h_j$, 
\item $\phi_3$ states that a process in a location $h_i$ will eventually be in a location $h_j$,
\item $\phi_4$ states that a process in a location $\ell_i$ will eventually be in a location $\ell_j$, and
\item $\phi_5$ states that a process in location $h_0$ will eventually be in location $h_0$.
\end{ienumerate}

\begin{table}[h]
\caption[position=above]{\footnotesize Comparison of verification times for different versions of \cref{fig:gtpProtocol} and properties $\phi_i$, based on the summary automaton ($\hat{A}$) or the cutoff system ($A^c$).}
\label{tab:performance-full}
\scriptsize
\begin{tabular}{ccc|cc|rr}
\toprule
\multicolumn{3}{c|}{\textbf{Benchmark}} &  & & \multicolumn{2}{c}{\textbf{Time}}  \\
Name & $\left| Q \right|$ & Cutoff $c$ & Property & Result & $\hat{A}$~~~&  $A^c$~\\
\midrule
\multirow{2}{*}{$GCS^\top(3)$} & \multirow{2}{*}{8} & \multirow{2}{*}{4}&$\phi_1, \phi_2$& False & ~~$< 0.1$ & ~~$< 0.1$  \\

  & &  &$\phi_3, \phi_4, \phi_5 $ & True  & $< 0.1$ & 26.0  \\
    
 \midrule
\multirow{2}{*}{$GCS(3)$} & \multirow{2}{*}{8}  & \multirow{2}{*}{7} &$\phi_1, \phi_2$ & False & $< 0.1$ & $< 0.1$ \\
      
  & & &$\phi_3, \phi_4, \phi_5$ & True &  $< 0.1$ & TO  \\
	
\midrule
\multirow{2}{*}{$GCS^\top(4)$} & \multirow{2}{*}{10} & \multirow{2}{*}{5}&$\phi_1, \phi_2$& False & ~~$< 0.1$ & ~~$< 0.1$  \\

  & &  &$\phi_3, \phi_4, \phi_5 $ & True  & $< 0.1$ & TO  \\
    
 \midrule
\multirow{2}{*}{$GCS(4)$} & \multirow{2}{*}{10}  & \multirow{2}{*}{9} &$\phi_1, \phi_2$ & False & $< 0.1$ & $< 0.1$ \\
      
  & & &$\phi_3, \phi_4, \phi_5$ & True &  $< 0.1$ & TO  \\

\bottomrule
\end{tabular}
\end{table}

	\fi

\end{document}